\documentclass[10pt, final, journal, letterpaper, twocolumn]{IEEEtran}
\makeatletter

\usepackage{bbm}
\usepackage{amsfonts}
\usepackage[dvips]{graphicx}
\usepackage{times}
\usepackage{cite}
\usepackage{amsmath}
\usepackage{array}
\usepackage{amssymb}

\usepackage{stfloats}
\usepackage{slashbox}
\usepackage{graphicx}
\usepackage{footnote}
\usepackage{booktabs}
\usepackage{array}
\usepackage{algorithmic}
\usepackage{algorithm}
\usepackage{subeqnarray}
\usepackage{cases}
\usepackage{threeparttable}
\usepackage{color}
\usepackage{hyperref}
\usepackage{epstopdf}
\usepackage{bm}
\usepackage{multirow}
\usepackage[labelformat=simple]{subcaption}
\usepackage{adjustbox}
\usepackage{color}

\newtheorem{theorem}{Theorem}

\newtheorem{lemma}{Lemma}

\newtheorem{corollary}{Corollary}

\allowdisplaybreaks[4]

\begin{document}
\abovedisplayskip=2pt
\belowdisplayskip=2pt
\title{Modeling, Analysis, and Optimization of Caching in Multi-Antenna Small-Cell Networks}
\author{\authorblockN{Xianzhe Xu and Meixia Tao}\\
\thanks{X. Xu and M. Tao are with the Department of Electronic Engineering, Shanghai
Jiao Tong University, Shanghai, China (Emails: august.xxz@sjtu.edu.cn;
mxtao@sjtu.edu.cn). Part of this work was presented in IEEE Globecom 2017 \cite{xxz_c}. This work is supported by the National Natural Science Foundation of China under grant 61571299 and the Shanghai Key Laboratory Funding under grant STCSM18DZ2270700.}
}

\maketitle
\vspace{-1.5cm}
\begin{abstract}
In traditional cache-enabled small-cell networks (SCNs), a user can suffer strong interference due to content-centric base station association. This may degenerate the advantage of collaborative content caching among multiple small base stations (SBSs), including probabilistic caching and coded caching. In this work, we tackle this issue by deploying multiple antennas at each SBS for interference management. Two types of beamforming are considered. One is matched-filter (MF) to strengthen the effective channel gain of the desired signal, and the other is zero-forcing (ZF) to cancel interference within a selected SBS cooperation group. We apply these two beamforming techniques in both probabilistic caching and coded caching, and conduct performance analysis using stochastic geometry. We obtain exact and approximate compact integral expressions of system performances measured by average fractional offloaded traffic (AFOT) and average ergodic spectral efficiency (AESE). Based on these expressions, we then optimize the caching parameters for AFOT or AESE maximization. For probabilistic caching, optimal caching solutions are obtained. For coded caching, an efficient greedy-based algorithm is proposed. Numerical results show that multiple antennas can boost the advantage of probabilistic caching and coded caching over the traditional most popular caching with the proper use of beamforming.
\end{abstract}
\section{Introduction}
In recent decades, mobile data traffic has experienced an explosive growth due to the rapid development of smart devices. Caching popular files in small base stations (SBSs) during off-peak time is a promising way to alleviate the peak-time congestion and avoid repetitive backhaul transmissions in wireless networks\cite{bastug2014living,shanmugam2013femtocaching,
wang2014cache,golrezaei2013femtocaching,barish2000world,7485844,7537172}. Many recent works have studied cache-enabled wireless networks with various caching strategies and performance metrics. These cache strategies can be broadly divided into two categories, uncoded caching and coded caching. In uncoded caching, a file is cached either entirely or not at all without partitioning in each cache-enabled SBS. Uncoded caching further includes deterministic caching \cite{bacstu?2015cache,7194828,peng2015backhaul} and probabilistic caching \cite{blaszczyszyn2015optimal,chae2016caching,7723871,7562510,7365479,7435255,bharath2016learning,8125744}. One typical deterministic cache strategy is the most popular caching (MPC), where each SBS only caches the most popular files until its cache size is full \cite{bacstu?2015cache,7194828}. Compared with MPC, probabilistic caching with optimized probabilities can achieve higher cache hit probability and higher successful transmission probability \cite{blaszczyszyn2015optimal,chae2016caching}. In coded caching, each file is first partitioned into multiple segments, these segments after coding are then cached in different SBSs \cite{bioglio2015optimizing,7572987,altman2013coding,chen2016cooperative,xu2017modeling,8017548}. In particular, the maximum distance separable (MDS) code is utilized in \cite{bioglio2015optimizing,7572987} and random linear network coding (RLNC) is applied in \cite{altman2013coding}. The work \cite{xu2017modeling} shows that coded caching outperforms probabilistic caching and MPC over a wide range of system parameters, but they all converge together when the content popularity is highly {skewed} or the decoding threshold at each user receiver is high.
The aim of this work is to investigate the role of multiple antennas for interference management in cache-enabled small-cell networks (SCNs) for further exploiting the advantage of collaborative caching.

Several works have studied the caching design in multi-antenna wireless networks. The works \cite{7136382,7488289} focus on the physical-layer optimization of base station clustering and beamforming in cache-enabled multi-antenna networks, while their caching strategies are given and not optimized. In \cite{7936549,cao2018treating}, the authors investigate the fundamental limits of caching in MIMO interference networks and MIMO broadcast channel from information theory perspective. Therein, only small-scale fading is considered and thus the obtained MIMO gain cannot be directly extended to SCNs where distance-dependent large-scale fading is present. Using stochastic geometry, the authors in \cite{liu2017caching} optimize the caching probabilities in cache-enabled heterogeneous networks for maximizing the success probability and the area spectral efficiency, respectively, with zero-forcing (ZF) beamforming. However, {in \cite{liu2017caching} only macro base stations (MBSs) that store all files have multiple antennas while each cache helper is still equipped with a single antenna. As a result, when the typical user is associated with the MBS tier, it is always served by the nearest MBS since all files are cached at MBSs and hence no cache-induced interference management gain is exploited.} To our best knowledge, a comprehensive treatment of caching analysis and optimization in multi-antenna SCNs {that allows each user to be served by any multiple nearby multi-antenna SBSs with specific beamforming structures} is not available in the literature.

{Note that to facilitate the optimization of cache strategies, one needs to obtain the coverage probability expressions. The works \cite{jindal2011multi,7038201,lee2015spectral} analyze the performance of multi-antenna networks by using stochastic geometry. The authors in \cite{jindal2011multi} analyze the coverage probability of the receiver connecting to a transmitter with fixed distance in ad hoc networks when ZF beamforming is utilized. The works \cite{7038201,lee2015spectral} also obtain the coverage probability of the typical user when it is served by the nearest SBS with random distance in ZF beamforming case. Then they approximate and simplify the expressions due to the high complexity, which is a common issue in multi-antenna networks by using stochastic geometry. Comparing with these existing works, we analyze the coverage probability of the typical user when it can be served by its multiple nearby SBSs in both ZF and matched-filter (MF) beamforming cases. The approximation and simplification in our work are much more challenging since each user can be served by multiple nearby SBSs and the joint distribution of their distances is much more complicated. Moreover, the interference of MF comes from all SBSs except the serving one, which is more complicated than the ZF case.}

In this paper, we consider caching analysis and optimization in multi-antenna SCNs by taking two specific beamforming techniques into account. One is ZF where the multiple antennas in each SBS are used to cancel interference within a selected SBS cooperation group. The other is MF where the multiple antennas in each SBS are used to strengthen the effective channel gain of desired signals without SBS coordination. Utilizing tools from stochastic geometry, we analyze the performance of a typical user for both probabilistic caching and coded caching. {Due to the high complexity of the analytical expressions, we approximate and simplify the results and obtain good approximations with much simpler structures and lower computational complexity. Based on the obtained analytical or approximate results,} we then optimize the caching parameters. Our prior conference paper \cite{xxz_c} considered the analysis and optimization for probabilistic caching {with perfect channel state information (CSI)} only. The main contributions and results of this journal version are summarized as follows.

$\bullet$ We propose a user-centric SBS clustering and transmission framework where each user can only communicate with a certain number of nearby SBS. Specifically, in the probabilistic caching model, each user is associated with the nearest SBS within its cluster that has cached its requested file. The serving SBS adopts either ZF beamforming for intra-cluster interference cancellation if transmission coordination is allowed or MF beamforming otherwise. In the coded caching model, each user collects sufficient number of coded segments of its requested file from multiple SBSs within its cluster. These multiple serving SBSs can transmit sequentially in an orthogonal manner using ZF beamforming (O-ZF) if transmission coordination is allowed. Otherwise, they transmit concurrently in a non-orthogonal manner using MF beamforming (NO-MF) in conjunction with successive interference cancellation (SIC) at each user receiver.

$\bullet$ We obtain tractable expressions of the coverage probabilities for all the considered caching and beamforming schemes. {Due to the high computational complexity,} we {approximate and simplify the analytical expressions and} derive a set of more compact forms for the (approximate) coverage probability bounds. Based on these expressions, we derive approximate and compact integral expressions for the average fractional offloaded traffic (AFOT) and the average ergodic spectral efficiency (AESE). {We extend our analysis to the imperfect CSI case and obtain the corresponding expressions similarly.}

$\bullet$ In the probabilistic caching model, we formulate two optimization problems for the caching probabilities towards AFOT and AESE maximization, respectively. We show that these optimization problems are convex and also obtain the optimal solutions. In the coded caching model, we formulate a unified cache placement problem as a multiple-choice knapsack problem (MCKP) for AFOT and AESE maximization, respectively. {By analyzing and exploiting the properties of the problem, we propose a greedy-based low-complexity algorithm to solve this NP-hard problem, which is shown to perform almost the same as the optimal exhaustive search algorithm.}

$\bullet$  Numerical results reveal that both probabilistic caching and coded caching can enjoy a higher performance gain from multiple antennas than MPC by allowing the collaborative caching among SBSs. Numerical results also show that ZF beamforming performs better than MF when the number of antennas in each SBS is larger than the cluster size. When the number of antennas in each SBS is the same as the cluster size, MF outperforms ZF in most cases. {Moreover, MF is more robust than ZF when SBSs obtain quantized CSI via limited feedback.}

The rest of the paper is organized as follows. The system model is introduced in Section \uppercase\expandafter{\romannumeral2}. In Section \uppercase\expandafter{\romannumeral3} and Section \uppercase\expandafter{\romannumeral4}, we analyze AFOT and AESE with tools from stochastic geometry and optimize the caching strategy in probabilistic caching and coded caching, respectively. {We then extend the analysis and optimization of caching to the imperfect CSI case in Section \uppercase\expandafter{\romannumeral5}. The numerical results are presented in Section \uppercase\expandafter{\romannumeral6}. Finally, we conclude the paper in Section \uppercase\expandafter{\romannumeral7}.}

{\emph{Notation}}: This paper uses bold-face lower-case $\mathbf{h}$ for vectors and bold-face uppercase $\mathbf{H}$ for matrices. $\mathbf{H}^H$ is the conjugate transpose of $\mathbf{H}$ and $\mathbf{H}^\dagger$ is the left pseudo-inverse of $\mathbf{H}$. $\mathbf{I}_m$ implies the $m\times m$ identity matrix and $\mathbf{0}_{1\times m}$ denotes the $1\times m$ zero vector.
\section{System Model}
We consider a cache-enabled multi-antenna SCN, where each SBS is equipped with $L$ transmit antennas as well as a local cache, and is located on a two-dimensional plane according to a homogenous Poisson point process (HPPP), denoted as $\Phi_b=\{\textbf{d}_i\in\mathbb{R}^2,\forall i\in\mathbb{N^+}\}$ with intensity $\lambda_b$. Each user has a single receive antenna and their locations are modeled as another independent HPPP with intensity $\lambda_u$. It is assumed that $\lambda_u\gg\lambda_b$ so that the network is fully loaded with all the SBSs being active at any given time instant. Each user can only choose its serving SBS or SBSs from a cluster of $K$ nearest SBSs to limit strong interference, where $K\geq2$. {We refer to the $K$ nearest SBSs of each user as a user-centric SBS cluster with size $K$, which is formed by the central controller}. As such, the plane is tesselated into $K$-th order Voronoi cells{\cite{lee2015spectral}}. The $K$-th order Voronoi cell associated with a set of $K$ points $\textbf{d}_1,\cdots,\textbf{d}_K$  is the region that all the points in this region are closer to these $K$ points than to any other point of $\Phi_b$, i.e., $\mathcal{V}_K(\textbf{d}_1,\cdots,\textbf{d}_K)=\{\textbf{d}\in\mathbb{R}^2|\cap_{k=1}^K \{\Vert \textbf{d}-\textbf{d}_k\Vert \leq \Vert \textbf{d}-\textbf{d}_i\Vert\},\textbf{d}_i\in \Phi_b\backslash\{\textbf{d}_1,\textbf{d}_2,\cdots,\textbf{d}_K\}\}$.

Without loss of generality, we focus on a typical user, denoted as $u_0$, located at the origin. The typical user can connect to any of the SBSs in the $K$-th order Voronoi cell that it belongs to, where the set of SBSs is denoted as $\Phi_K=\{\textbf{d}_1,\textbf{d}_2,\cdots,\textbf{d}_K\}$. The distance between $u_0$ and the $k$-th nearest SBS $\textbf{d}_k$ is denoted as $r_k$, satisfying $0<r_1 \leq r_2 \ldots \leq r_K$. The channel between $u_0$ and every SBS $\textbf{d}_i \in \Phi_b$ in the network consists of both Rayleigh-distributed small-scale fading, denoted as $\mathbf{h}_{i0}\in \mathbb{C}^{1 \times L}$ with $\mathbf{h}_{i0}\sim \mathcal{CN}(\mathbf{0}_{1\times{L}},\mathbf{I}_{L})$, and distance-dependent large-scale fading, denoted as $r_i^{-\frac{\alpha}{2}}$, with path-loss exponent $\alpha>2${\cite{andrews2011tractable}}.
\subsection{Cache Placement Model}
We consider a file library $\mathcal{F}=\{f_1,f_2,\cdots,f_N\}$, where $N$ is the total number of files. All files are assumed to have the same normalized size of $1$. The popularity of file $f_n$ is denoted as $p_n$, satisfying $0\leq p_n\leq 1$ and $\sum_{n=1}^Np_n=1$. Without loss of generality, we assume $p_1\geq p_2\ldots\geq p_N$. Each SBS can store up to $M$ files with $M<N$ to avoid the trivial case.

Two cache placement models are considered in this paper.
\subsubsection{Probabilistic Caching}
In the \emph{probabilistic caching} model, each SBS caches file $f_n$ with probability $a_n$. Due to the cache size constraint and probability property, we have the constraints: $\sum_{n=1}^N a_n\leq M$ and $0\leq a_n \leq 1$ for $n=1,2,\ldots,N$ \footnote{These are sufficient and necessary conditions for the existence of a probabilistic cache placement scheme meeting strictly the cache size constraint at each SBS. A practical placement approach can be found in \cite{blaszczyszyn2015optimal}}. We denote the cache strategy as a $1\times N$ vector $\textbf{a}=[a_1,a_2,\ldots,a_N]$. With such probabilistic caching strategy, when the typical user $u_0$ requests for file $f_n$, it will be associated with the nearest SBS that caches $f_n$ in $\Phi_K$. If the transmission fails or none of the SBSs in $\Phi_K$ has cached $f_n$, $u_0$ will be served by a macro base station (MBS) that can download the requested file from the core network at a much higher cost.
\subsubsection{Coded Caching}
In the \emph{coded caching} model, each file $f_n$ is split into $b_n$ disjoint fragments with size $\frac{1}{b_n}$ for each, where $b_n\in\mathcal{K}\cup\left\{\infty\right\}$ with $\mathcal{K}\triangleq\left\{1,2,\ldots,K\right\}$. These fragments are then encoded into a large number of coded packets\footnote{{Note that the coded caching in this work is different from the coded caching scheme proposed in \cite{6763007}, where the cached contents are carefully designed to allow serving users via multicast transmissions.}}. {Specifically, an MDS $(a,b)$ code is to generate $a$ encoded packets from $b$ input fragments, such that any subset of $b$ encoded packets is sufficient to recover the data. There are some well-known examples of MDS codes, such as Reed Solomon codes.} Following the convention \cite{xu2017modeling}, we do not restrict to any specific coding scheme {in this work} and only require that each file $f_n$ can be successfully recovered from any set of $b_n$ coded packets. {In the placement phase, each SBS caches one and only one distinct coded packet for each file $f_n$ and the user can recover the original file if it receives any $b_n$ coded packets} \footnote{Similar to \cite{xu2017modeling}, the distinction between the coded packets stored in the SBSs within each $K$-th order Voronoi cell can be ensured with a graph-coloring approach if MDS codes are used, or ensured with large probability if RLNC is used.}. For the special case when $b_n=\infty$, the file $f_n$ is not cached in any SBS. When $b_n=1$, the file $f_n$ is cached entirely in each SBS. We denote the cache strategy as a $1 \times N$ vector $\textbf{b}=[b_1,b_2,\ldots,b_N]$ with cache size constraint $\sum_{n=1}^N \frac{1}{b_n} \leq M$. With such coded caching strategy, when the typical user $u_0$ requests for file $f_n$, it needs to collect $b_n$ coded packets of $f_n$ from the $b_n$ nearest SBSs in the cluster $\Phi_K$. These $b_n$ SBSs will transmit to $u_0$ either sequentially in an orthogonal manner or concurrently in a non-orthogonal manner as detailed in the next subsection. If $u_0$ fails to collect enough coded packets from the SBSs in $\Phi_K$ due to transmission error or $b_n=\infty$, it will acquire the missing coded packets from a MBS at a much higher cost.
\subsection{Transmission Model}
As assumed earlier, we consider a fully loaded system so that all the SBSs in the network are active at any time. Throughout this paper, we consider an interference-limited network where the noise can be neglected and use signal-to-interference ratio (SIR) for performance analysis. At any time, the received signal of $u_0$ is given by (ignoring noise):
\begin{equation}
y_0=\sum_{\textbf{d}_i\in\Phi_b}r_i^{-\frac{\alpha}{2}}\mathbf{h}_{i0}\mathbf{w}_{i}x_{i},   \label{eqn:receive_signal}
\end{equation}
where $x_i$ denotes the transmit signal from SBS $\textbf{d}_i$ and $\mathbf{w}_i\in \mathbb{C}^{L \times 1}$ denotes the corresponding beamforming vector. The intended receiver of $x_i$ and the design of $\mathbf{w}_i$ depend on the requested file $f_n$ and its cache parameter $a_n$ (for probabilistic caching) or $b_n$ (for coded caching).

{There are many different types of beamformings, such as MMSE, which can optimally balance signal boosting and interference cancellation and maximizes the SIR. We consider ZF and MF in this work since they are generally more amenable to analysis than MMSE because of their simple structures \cite{jindal2011multi}. Besides, ZF and MF have a distinct advantage in terms of implementation complexity compared to MMSE \cite{7097743}.}


\subsubsection{Transmission with probabilistic caching}
In the probabilistic caching model, we assume that SBS $\textbf{d}_k \in \Phi_K$, for $k=1,2,\ldots,K$, is the nearest SBS that has cached the requested file $f_n$ and therefore $u_0$ is associated with $\textbf{d}_k$ during the transmission. We consider two types of transmit beamforming at each SBS. One is uncoordinated, where each SBS applies an MF based beamforming independently to maximize the effective channel
gain of its own user. The beamforming vector $\mathbf{w}_{k,\text{mf}}$ of SBS $\textbf{d}_k$ serving $u_0$ is given by:
\begin{equation}
\mathbf{w}_{k,\text{mf}}=\frac{\mathbf{h}_{k0}^H}{\Vert\mathbf{h}_{k0}\Vert}. \label{eqn:MF-vector}
\end{equation}
{Here, we assume that each user estimates the downlink CSI from its serving SBS and then conveys the CSI back to the SBS via errorfree feedback links.} Since each SBS serves its own user independently, the interference observed by $u_0$ comes from all SBSs except $\textbf{d}_k$ in the network. Thus, the SIR of $u_0$ when served by SBS $\textbf{d}_k$ with MF beamforming is given by:
\begin{equation}
\text{SIR}_{k,\text{mf}}=\frac{g_{k,\text{mf}}\cdot r_k^{-\alpha}}{\sum_{\textbf{d}_j\in\Phi_b \backslash
\{\textbf{d}_k\}} g_{j,\text{mf}}\cdot r_j^{-\alpha}},         \label{eqn:mf_SIR}
\end{equation}
where $g_{k,\text{mf}}=\Vert\mathbf{h}_{k0}\Vert^2$ is the effective channel gain of the desired signal from $\textbf{d}_k$ and follows the Gamma distribution with shape parameter $L$ and scale parameter $1$, denoted as $g_{k,\text{mf}}\sim \Gamma(L,1)$, and $g_{j,\text{mf}}=|\mathbf{h}_{j0}\mathbf{w}_{j,\text{mf}}|^2$ is the effective channel gain of the undesired signal from $\textbf{d}_j$ and follows the exponential distribution with mean $1$, denoted as $g_{j,\text{mf}}\sim \exp(1)$ \cite{jindal2011multi}.

The other type of beamforming is coordinated, where the $K$ SBSs in each $K$-th order Voronoi cell apply ZF beamforming so as to null out the intra-cluster interference. As such, all the SBSs in the $K$-th order Voronoi cell $\mathcal{V}_K(\textbf{d}_1,\cdots,\textbf{d}_K)$ that the typical user $u_0$ falls into can simultaneously serve $K$ users (including $u_0$) located in the same $\mathcal{V}_K(\textbf{d}_1,\cdots,\textbf{d}_K)$ without intra-cluster interference \cite{lee2015spectral}. We assume that $L \geq K$ to ensure the feasibility of ZF beamforming. Due to the assumption that $\lambda_u \gg \lambda_b$, there will always exist $K$ users in the Voronoi cell with each choosing a distinct SBS to communicate with during the delivery phase. The beamforming vector of SBS $\textbf{d}_k$ serving $u_0$ is given by\cite{7038201}:
\begin{equation}
\mathbf{w}_{k,\text{zf}}=\frac{\mathbf{(I}_{L}-\mathbf{H}\mathbf{H}^\dagger)\mathbf {h}_{k0}^T}{\Vert\mathbf{(I}_{L}-\mathbf{H}\mathbf{H}^\dagger)\mathbf{h}_{k0}^T\Vert},\label{eqn:zf}
\end{equation}
where $\mathbf{H}=[\mathbf{h}_{k1}^T,\mathbf{h}_{k2}^T,\cdots,\mathbf{h}_{k(K-1)}^T]$ is the channel between the serving SBS $\textbf{d}_k$ of the typical user and the $K-1$ users served by the other $K-1$ SBSs within the cluster. {Here, we assume that each user estimates the downlink CSI from the $K$ nearest SBSs within its cluster by means of orthogonal pilot symbols and then conveys the CSI back to the SBSs via errorfree feedback links.} By (\ref{eqn:zf}), the $K-1$ intra-cluster interference in the cell can be completely nulled when $L \geq K$ and the interference observed by $u_0$ only comes from the SBSs out of the cluster, i.e., $\Phi_b\backslash \Phi_K$. Therefore, the SIR of $u_0$ when served by SBS $\textbf{d}_k$ with ZF beamforming is given by:
\begin{equation}
\text{SIR}_{k,\text{zf}}=\frac{g_{k,\text{zf}}\cdot r_k^{-\alpha}}{\sum_{\textbf{d}_j\in\Phi_b\backslash \Phi_K
}g_{j,\text{zf}}\cdot r_j^{-\alpha}}, \label{eqn:zf-SIR}
\end{equation}
where $g_{k,\text{zf}}=|\mathbf{h}_{k0}\mathbf{w}_{k,\text{zf}}|^2$ is the effective channel gain of the desired signal and follows $g_{k,\text{zf}}\sim \Gamma(L-K+1,1)$, and $g_{j,\text{zf}}=|\mathbf{h}_{j0}\mathbf{w}_{j,\text{zf}}|^2$ is the effective channel gain of the undesired signal and follows $g_{j,\text{zf}}\sim \exp(1)$ \cite{jindal2011multi}.

\subsubsection{Transmission with coded caching}
We assume that file $f_n$ is split into $1\leq b_n \leq K$ fragments for coded caching and thus $u_0$ is associated with the $b_n$ nearest SBSs in $\Phi_K$ during the transmission. Similar to the previous case, we consider both ZF and MF beamforming at each SBS. But the specific design differs due to that the typical user $u_0$ needs to receive signals from multiple SBSs rather than just one SBS for content delivery.

In the NO-MF scheme, the nearest $b_n$ SBSs use non-orthogonal transmission to deliver their cached coded packets of $f_n$ to $u_0$ concurrently at the same resource block. The user adopts SIC to decode the signals successively in the descending order of the average received signal strength, or equivalently, from the nearest to the farthest in the considered  homogeneous SCNs. Specifically, the signal from SBS $\textbf{d}_1$ is decoded first and, if successful, subtracted from the received signal, then the algorithm proceeds to the signal from SBS $\textbf{d}_2$, and so on. Note that when the user decodes the signal from the $k$-th nearest SBS $\textbf{d}_k$, the interference comes from the SBSs farther than SBS $\textbf{d}_k$, i.e., $\Phi_b\backslash \Phi_k$, where $\Phi_k=\{\textbf{d}_1,\textbf{d}_2,\ldots,\textbf{d}_k\}$. In addition, the effective channel gains of the signals coming from SBS $\textbf{d}_{k+1}$ to SBS $\textbf{d}_{b_n}$ all follow the Gamma distribution with shape parameter $L$ since they are the serving SBSs and their beamformers are matched to the channel of $u_0$. The rest interference channel gains from the SBSs farther than $\textbf{d}_{b_n}$ follow exponential distribution with unit mean. Given the MF beamforming vector $\mathbf{w}_{k,\text{mf}}$ in (\ref{eqn:MF-vector}), the SIR of $u_0$ for decoding the signal from $\textbf{d}_k$ is given by:
\begin{equation}
\text{SIR}_{k,\text{no-mf}}=\frac{g_{k,\text{no-mf}}\cdot r_k^{-\alpha}}{\sum_{\textbf{d}_i\in \Phi_{\textbf{d}_{b_n}}\!\!\!\backslash\Phi_k}g_{i,\text{no-mf}}\cdot r_i^{-\alpha}\!\!+\!\!\sum_{\textbf{d}_j\in\Phi_b \backslash
\Phi_{\textbf{d}_{b_n}}} \!\!\!\!g_{j,\text{no-mf}}\cdot r_j^{-\alpha}}, \label{eqn:SIR-NO}
\end{equation}
where $\Phi_{\textbf{d}_{b_n}}=\{\textbf{d}_1,\textbf{d}_2,\ldots,\textbf{d}_{b_n}\}$ is the set of nearest $b_n$ SBSs of $u_0$, $g_{k,\text{no-mf}}\sim \Gamma(L,1)$, $g_{i,\text{no-mf}}\sim \Gamma(L,1)$ for SBSs $\textbf{d}_i\in \Phi_{\textbf{d}_{b_n}}\backslash\Phi_k$ and $g_{j,\text{no-mf}}\sim \exp(1)$ for SBSs $\textbf{d}_j\in\Phi_b \backslash\Phi_{\textbf{d}_{b_n}}$.

In the O-ZF scheme, the nearest $b_n$ SBSs take turns to deliver their cached coded packets to $u_0$ in an orthogonal manner. Note that when SBS $\textbf{d}_k$, for $1\leq k\leq b_n$, is serving $u_0$, the other $K-1$ SBSs in the cluster $\Phi_K$ are serving simultaneously other $K-1$ users in the same $K$-th order Voronoi cell. Same to the probabilistic caching, we adopt the ZF beamforming in (\ref{eqn:zf}). The received SIR of $u_0$ is the same as (\ref{eqn:zf-SIR}), for $k=1,2,\ldots,b_n$.
\subsection{Performance Metrics}
In this paper, we adopt two performance metrics to measure the gain brought by caching in different caching models and transmission schemes.
\subsubsection{AFOT}
The AFOT, denoted as $\overline{L}(K)$, measures the average fraction of successfully offloaded traffic from cache-enabled SBSs. The traffic offload is said to be {\emph{successful}} if the requested file is cached locally in SBSs and the corresponding received SIR is above a certain decoding threshold {$\gamma$}. Let $L_n$ denote the fractional offloaded traffic (FOT) given that $u_0$ requests file $f_n$. The AFOT is given by:
\begin{align}
\overline{L}(K,{\gamma})=\sum_{n=1}^{N}p_nL_n.       \label{eqn:AFOT}
\end{align}
The specific definition of $L_n$ shall be introduced later in Section \uppercase\expandafter{\romannumeral 3} for probabilistic caching and in Section \uppercase\expandafter{\romannumeral 4} for coded caching. {Note that AFOT characterizes the traffic offloading capability of the cache-enabled SBSs. Namely, it captures how likely or at what fraction the typical user can download its requested file from the cache-enabled SBSs locally at a given target transmission rate without resorting to the core network. It depends on the cache policy of each SBS and the target transmission rate to the typical user.}
\subsubsection{AESE}
The AESE, denoted as $\overline{S}(K)$, measures the average ergodic spectral efficiency of each cache-enabled SBS when serving a typical user. Let $S_n$ denote the ergodic spectral efficiency (ESE) given that $u_0$ requests file $f_n$. The AESE is given by:
\begin{align}
\overline{S}(K)=\sum_{n=1}^{N}p_nS_n.   \label{eqn:AESE}
\end{align}
The specific definition of $S_n$ shall be introduced later in Section \uppercase\expandafter{\romannumeral 3} for probabilistic caching and Section \uppercase\expandafter{\romannumeral 4} for coded caching.

{From the above definitions, AFOT can be used to measure the service performance for delay-sensitive content requests, such as video on-demand, that require a target minimum transmission rate regardless of network condition; AESE can be used to measure the service performance for other delay-insensitive content requests, such as file download, where maximizing the average download rate is desired. Similar performance metrics are also both considered in \cite{xu2017modeling,liu2017caching}.}



\section{Analysis and Optimization of Probabilistic Caching}
In this section, we analyze AFOT and AESE in the probabilistic caching model. First, we analyze the coverage probability of the typical user for different transmission schemes. Then based on the results of the coverage probability, we derive and analyze FOT and ESE, respectively. Finally, we optimize the cache vector $\textbf{a}$ by maximizing the AFOT and AESE, respectively.
\subsection{Coverage Probability}
We first analyze the coverage probability of the typical user $u_0$ when it is served by the $k$-th nearest SBS. It is defined as the probability that the received SIR exceeds a given SIR target $\gamma$:
\begin{align}
P_{\text{cov}}^k(K,\gamma)=P[\text{SIR}_k\geq\gamma],    \label{eqn:SIR}
\end{align}
where $\text{SIR}_k$ is given in (\ref{eqn:mf_SIR}) for MF beamforming or (\ref{eqn:zf-SIR}) for ZF beamforming.
\subsubsection{MF Beamforming}
\begin{lemma}[{Coverage Probability of MF}] \label{lemma:1}
The coverage probability of the typical user served by the $k$-th nearest SBS with MF beamforming, for $k=1,2,\ldots,K$, is given by:
\begin{align}
P_{\text{cov,mf}}^{k}(K,\gamma)=\mathbb{E}_{r_k}\left[\sum_{i=0}^{L-1}\frac{(-\gamma {r_k}^\alpha)^i}{i!}\mathcal{L}_{I_{r1}}^{(i)}(\gamma {r_k}^\alpha)|r_k\right],\label{eqn:coveragexyz}
\end{align}
where $I_{{r1}}=\sum_{\textbf{d}_j\in\Phi_b\backslash\{\textbf{d}_k\}}g_{j,\text{mf}}\cdot r_j^{-\alpha}$, $\mathcal{L}_{I_{r1}}(s)=\mathbb {E}[e^{-sI_{r1}}]=\mathcal{L}_{I_1}(s)\cdot\mathcal{L}_{I_2}(s)$ is the Laplace transform of $I_{{r1}}$, where $\mathcal{L}_{I_1}(s)$ and $\mathcal{L}_{I_2}(s)$ are given in (\ref{eqn:L1}) and (\ref{eqn:L2}), respectively, and $\mathcal{L}_{I_{r1}}^{(i)}(s)$ is the $i$-th order derivative of $\mathcal{L}_{I_{r1}}(s)$.

\end{lemma}

\begin{proof}
See Appendix A.
\end{proof}

The expression (\ref{eqn:coveragexyz}) takes the expectation over $r_k$. The probability density function (pdf) of $r_k$ is $f_{R_k}(r_k)=\frac{2\left(\lambda_b\pi r_k^2\right)^k}{r_k\Gamma(k)}\exp\left(-\lambda_b \pi r_k^2\right)$ \cite{haenggi2005distances},
where $\Gamma(k)=(k-1)!$ is the Gamma function.

Since the tractable expression of the coverage probability is complex, we provide more compact forms to bound the coverage probability in the following theorem.

\begin{theorem}[{Bound of Coverage Probability with MF}] \label{theorem:MF bound}
The coverage probability of the typical user served by the $k$-th nearest SBS with MF beamforming, for $k=1,2,\ldots,K$, is bounded as:
\begin{align}
P_{\text{cov,mf}}^{k,\text{l}}(K,\gamma) \leq P_{\text{cov,mf}}^{k}(K,\gamma) \leq P_{\text{cov,mf}}^{k,\text{u}}(K,\gamma),
\end{align}
with
\begin{align}
&P_{\text{cov,mf}}^{k,\text{u}}(K,\gamma)= \sum_{l=1}^{L} \beta_1\left(\eta,\gamma,\alpha,l,k\right)\frac{\binom{L}{l}(-1)^{l+1}}{\left(1+\beta_2\left(\eta,\gamma,\alpha,l\right)\right)^k},\label{eqn:app1} \\
&P_{\text{cov,mf}}^{k,\text{l}}(K,\gamma)=\sum_{l=1}^{L} \beta_1(1,\gamma,\alpha,l,k)\frac{\binom{L}{l}(-1)^{l+1}}{(1+\beta_2(1,\gamma,\alpha,l))^k}, \label{eqn:app11}
\end{align}
where
\begin{align}
&\beta_1(x,\gamma,\alpha,l,k)\!\!=\!\!\left[1-\frac{2 (x \gamma l)^{\frac{2}{\alpha}}}{ \alpha}B\left(\frac{2}{\alpha},1-\frac{2}{\alpha},\frac{1}{1+x \gamma l}\right)\right]^{k-1}, \label{eqn:beta1}\\
&\beta_2(x,\gamma,\alpha,l)=2\frac{\left(x \gamma l\right)^{\frac{2}{\alpha}}}{\alpha} B^{'}\left(\frac{2}{\alpha},1-\frac{2}{\alpha},\frac{1}{1+x \gamma l}\right),\label{eqn:beta2}
\end{align}
where $\eta=(L!)^{-\frac{1}{L}}$, $B(x,y,z)\triangleq \int_0^z u^{x-1}(1-u)^{y-1}du$ is the incomplete Beta function and $B^{'}(x,y,z)\triangleq \int_z^1 u^{x-1}(1-u)^{y-1}du$ is the complementary incomplete Beta function.
\end{theorem}

\begin{proof}
See Appendix B.
\end{proof}

In the special case with $L=1$ (single antenna), the upper and lower bounds coincide and hence give the exact expression of the coverage probability. In addition, if $\alpha=4$, the exact coverage probability can be written in a closed form.

\begin{corollary}
The coverage probability of the typical user $u_0$ served by the $k$-th nearest SBS in the single-antenna network with $\alpha=4$ is given by:
\begin{align}
&P_{\text{cov,mf}}^k(K,\gamma)=\frac{\left(1-\sqrt{\gamma}\text{arcsin}\frac{1}{\sqrt{1+\gamma}}\right)^{k-1}}{\left(1+\sqrt{\gamma}\text{arccos}\frac{1}{\sqrt{1+\gamma}}\right)^k}.\label{eqn:special1}
\end{align}
\end{corollary}

In the special case, when the user is served by its nearest SBS, i.e., $k=1$, Corollary $1$ reduces to the result given in \cite[Theorem $2$]{andrews2011tractable}.
\subsubsection{ZF Beamforming}
\begin{lemma}[{Coverage Probability of ZF}]
The coverage probability of the typical user served by the $k$-th nearest SBS with ZF beamforming, for $k=1,2,\ldots,K$, is given by:
\begin{align}
P_{\text{cov,zf}}^{k}(K,\gamma)=\mathbb{E}_{r_k,r_K}\left[\sum_{i=0}^{L-K}\frac{(-\gamma {r_k}^\alpha)^i}{i!}\mathcal{L}_{I_{r2}}^{(i)}(\gamma {r_k}^\alpha)|r_k,r_K\right], \label{eqn:cov-zf}
\end{align}
where $I_{r2}=\sum_{\textbf{d}_j\in\Phi_b\backslash \Phi_K}g_{j,\text{zf}}\cdot r_j^{-\alpha}$ and its Laplace transform is given in (\ref{eqn:L_2}).
\end{lemma}

The proof of lemma $2$ is similar to Appendix A, so we omit it here. {Notice that the expectation in (\ref{eqn:cov-zf}) is not only over $r_k$, but also over $r_K$ since the inter-cluster interference comes from the SBSs farther than $r_K$.} Thus, we need to know the joint pdf of $r_k$ and $r_K$, which is given by\cite{haenggi2005distances,srinivasa2010distance}:
\begin{align}
f_{R_k,R_K}(r_k,r_K)&=\frac{4(\lambda_b\pi)^K}{\Gamma(K-k)\Gamma(k)}r_kr_K(r_k^2)^{k-1}  \nonumber\\
&\times(r_K^2-r_k^2)^{K-k-1}\exp\left(-\lambda_b \pi r_K^2\right). \label{eqn:joint pdf}
\end{align}

%
%

Similarly, the tractable expression of the coverage probability is complex, we obtain more compact forms of the approximate coverage probability bounds in the following theorem.
\begin{theorem}[{Approximate Bound of Coverage Probability with ZF}]\label{theorem:BF bound}
The coverage probability of the typical user served by the $k$-th nearest SBS with ZF beamforming, for $k=1,2,\ldots,K$, can be approximately bounded as:
\begin{align}
P_{\text{cov,zf}}^{k,\text{l}}(K,\gamma) \lesssim P_{\text{cov,zf}}^{k}(K,\gamma) \lesssim P_{\text{cov,zf}}^{k,\text{u}}(K,\gamma),
\end{align}
with
\begin{align}
&P_{\text{cov,zf}}^{k,\text{u}}(K,\gamma)=\sum_{l=1}^{L-K+1} \frac{\binom{L-K+1}{l}(-1)^{l+1}}{{\left[1+\left(\kappa \gamma l\right)^{\frac{2}{\alpha}}\sqrt{\frac{k}{K}} \mathcal{A}\left(\frac{\sqrt K \left(\kappa \gamma l\right)^{-\frac{2}{\alpha}}}{\sqrt k }\right)\right]^k}},\label{eqn:app2}  \\
&P_{\text{cov,zf}}^{k,\text{l}}(K,\gamma)=\sum_{l=1}^{L-K+1} \frac{\binom{L-K+1}{l}(-1)^{l+1}}{{\left[1+\left( \gamma l\right)^{\frac{2}{\alpha}}\sqrt{\frac{k}{K}} \mathcal{A}\left(\frac{\sqrt K \left(\gamma l\right)^{-\frac{2}{\alpha}}}{\sqrt k }\right)\right]^k}}, \label{eqn:app22}
\end{align}
where $\mathcal{A}(x)=\int_x^\infty \frac{1}{1+u^{\frac{\alpha}{2}}}du$ and $\kappa=(L-K+1!)^{-\frac{1}{L-K+1}}$.
\end{theorem}

\begin{proof}
See Appendix C.
\end{proof}

The approximate upper and lower bounds coincide when $L=K$, and hence give the approximate coverage probability. Furthermore, when $\alpha=4$, the approximate coverage probability can be written in a closed form.

\begin{corollary}
The approximate coverage probability of the typical user served by the $k$-th nearest SBS in the ZF scheme with $L=K$ and $\alpha=4$ is given by:
\begin{align}
P_{\text{cov,zf}}^k(K,\gamma)\simeq \frac{1}{\left[1+\sqrt{\frac{k\gamma}{K}}\text{arccot}\left(\frac{K}{k\gamma}\right)\right]^k}.\label{eqn:special2}
\end{align}
\end{corollary}

When the user is served by the nearest SBS, i.e., $k=1$, Corollary $2$ reduces to the results given in \cite[Eqn. (28)]{lee2015spectral}.

{The above bounds in compact forms can be used as approximate coverage probability expressions for cache placement optimization at large $K$ and $L$ since the complexity of computing the exact analytical expressions increases rapidly as $K$ and $L$ increase. As we shall demonstrate numerically in Section \uppercase\expandafter{\romannumeral6}-A, the (approximate) upper bounds (\ref{eqn:app1}) and (\ref{eqn:app2}) are tighter than the (approximate) lower bounds (\ref{eqn:app11}) and (\ref{eqn:app22}). Nevertheless, we still use the (approximate) lower bounds in view of mathematical rigorousness since the optimization objectives, AFOT and AESE, are both increasing with respect to the coverage probabilities.}

\subsection{Fractional Offloaded Traffic}
Based on the {above} analysis of the coverage probability, {we analyze FOT $L_n$ in both MF and ZF schemes. Note that} when $u_0$ requests file $f_n$, it is served by the nearest SBS that caches $f_n$ within the cluster and the transmission is successful if the received SIR exceeds a given SIR target $\gamma$. Hence, the FOT $L_n$ is given by:
\begin{align}
L_n=\sum_{k=1}^K a_n(1-a_n)^{k-1}P_{\text{cov}}^{k}(K,\gamma), \label{eqn:T3}
\end{align}
{where $P_{\text{cov}}^{k}(K,\gamma)$ can be exact as (\ref{eqn:coveragexyz}) and (\ref{eqn:cov-zf}) or approximate as {(\ref{eqn:app11}) and (\ref{eqn:app22})} in both MF and ZF beamforming. Substituting (\ref{eqn:T3}) into (\ref{eqn:AFOT}), we then obtain the AFOT of probabilistic caching.}

\subsection{Ergodic Spectral Efficiency}
Based on the coverage probability, the ergodic achievable rate of the typical user served by the $k$-th nearest SBS is given by:
\begin{align}
R_{k}(K)&=\mathbb{E}\left[\text{log}_2(1+\text{SIR}_{k})\right] \nonumber\\
&=\int_0^\infty P\left[\text{log}_2(1+\text{SIR}_{k})>x\right]dx \nonumber\\
&=\int_0^\infty P_{\text{cov}}^{k}(K,2^x-1)dx.  \label{eqn:rate_mf_app}
\end{align}

By averaging all possible serving SBS $\textbf{d}_k\in \Phi_K$ when $u_0$ requests $f_n$, ESE is given by:
\begin{align}
S_n=\sum_{k=1}^K a_n(1-a_n)^{k-1}R_{k}(K), \label{eqn:S3}
\end{align}
where $R_{k}(K)$ is given in (\ref{eqn:rate_mf_app}). Substituting (\ref{eqn:S3}) into (\ref{eqn:AESE}), we then obtain the AESE.

\subsection{Caching Optimization}
In this section, we optimize the cache vector $\textbf{a}$ by maximizing the AFOT or AESE. {Note that we can use approximate coverage probabilities {(\ref{eqn:app11}) or (\ref{eqn:app22})} for cache placement optimization when $K$ and $L$ are large to maximize approximate AFOT or AESE.} The optimization problem can be formulated as:
\addtocounter{equation}{1}
\begin{align}
\text{\textbf{P1:}} \quad \underset{\textbf{a}}{\text{max}} \quad &\sum_{n=1}^N p_nQ_n, \tag{\theequation a}\\
\text{s.t} \quad &\sum_{n=1}^N a_n \leq M,   \tag{\theequation b}   \label{eqn:contraint1}  \\
 &0\leq a_n \leq 1, ~~~  n=1,2,\ldots,N,     \tag{\theequation c}
\end{align}
where $Q_n$ can be either the FOT $L_n$ in (\ref{eqn:T3}) or the ESE $S_n$ in (\ref{eqn:S3}) for both MF and ZF schemes. The constraint (\ref{eqn:contraint1}) can be rewritten as:
\begin{align}
\sum_{n=1}^N a_n=M,   \label{eqn:constraint}
\end{align}
without loss of optimality since caching more files increases the performance.
\begin{lemma}
The problem $\mathbf {P1}$ is convex for both MF and ZF schemes.
\end{lemma}

\begin{proof}
See Appendix D.
\end{proof}

By using KKT condition, the optimal solution of $\textbf{P1}$ satisfies the condition as follows.
\begin{theorem}
The optimal cache probabilities of $\textbf{P1}$ satisfy£º
\begin{align}
a_n(\mu^*)=\text{min}\left(1,w_n(\mu^*)\right),  \label{eqn:11111}
\end{align}
where $\mu^*\geq0$ is the optimal dual variable {to meet} the cache size constraint (\ref{eqn:constraint}) and $w_n(\mu^*)$ is the real and non-negative root of the equation:
\begin{align}
p_n\sum_{k=1}^K[1-w_n(\mu^*)]^{k-2}[1-kw_n(\mu^*)]P_{\text{cov}}^k(K,\gamma)-\mu^*=0, \label{eqn:sum}
\end{align}
for AFOT maximization, and the real and non-negative root of the equation:
\begin{align}
\!\!p_n\sum_{k=1}^K[1-w_n(\mu^*)]^{k-2}[1-kw_n(\mu^*)]R_k(K)-\mu^*=0,
\end{align}
for AESE maximization, respectively.
\end{theorem}

\begin{proof}
See Appendix E.
\end{proof}

To obtain the optimal cache vector $\textbf{a}$, we should find the optimal dual variable $\mu^*$ by substituting (\ref{eqn:11111}) into the cache size constraint $\sum_{n=1}^Na_n(\mu^*)=M$. From (\ref{eqn:strategy}) in Appendix E, it is observed that $a_n(\mu)$ is a decreasing function of $\mu$. Thus, the sum of $a_n(\mu)$ is also decreasing of $\mu$. Therefore, we can use the bisection method to find the optimal $\mu^*$.

\section{Analysis and Optimization of Coded Caching}
In this section, we analyze AFOT and AESE in the coded caching model. First, we analyze the coverage probability in different transmission schemes. Then based on the results of the coverage probability, we derive and analyze FOT and ESE, respectively. Finally, we optimize the cache vector $\textbf{b}$ by maximizing the AFOT and AESE, respectively.
\subsection{Coverage Probability}
Similar to (\ref{eqn:SIR}) for probabilistic caching, the coverage probability of the typical user when decoding the signal from the $k$-th nearest SBS out of the $b_n$ serving SBSs in coded caching is defined as the probability that the corresponding received SIR of $u_0$ exceeds a given SIR target $\gamma$. Specifically, for NO-MF transmission scheme, the coverage probability is given by:
\begin{align}
P_{\text{cov}}^{k}(b_n,\gamma)=P[\text{SIR}_k\geq\gamma],
\end{align}
where $\text{SIR}_k$ is given in (\ref{eqn:SIR-NO}) for $k=1,2,\ldots,b_n$. While for O-ZF transmission scheme, the coverage probability is the same as (\ref{eqn:SIR}) where $\text{SIR}_k$ is given in (\ref{eqn:zf-SIR}) for $k=1,2,\ldots,b_n$.
\subsubsection{NO-MF transmission}
\begin{lemma}[{Coverage Probability of NO-MF}]
The coverage probability of the typical user at SBS $\textbf{d}_k$ out of the $b_n$ serving SBSs with NO-MF transmission, for $k=1,2,\ldots,b_n$, is given by:
\begin{align}
P_{\text{cov,no-mf}}^{k}(b_n,\gamma)=\mathbb{E}_{r_k}\left[\sum_{i=0}^{L-1}\frac{(-\gamma {r_k}^\alpha)^i}{i!}\mathcal{L}_{I_{r3}}^{(i)}(\gamma {r_k}^\alpha)|r_k\right],\label{eqn:coverage}
\end{align}
where $I_{{r3}}\!\!=\!\!\sum_{\textbf{d}_i\in \Phi_{\textbf{d}_{b_n}}\backslash\Phi_k}\!\!g_{i,\text{no-mf}}\cdot r_i^{-\alpha}+\sum_{\textbf{d}_j\in\Phi_b \backslash\Phi_{\textbf{d}_{b_n}}}\!\!\! g_{j,\text{no-mf}}\cdot r_j^{-\alpha}$ and its Laplace transform is:
\begin{align}
\mathcal{L}_{I_{r3}}(s)&=\left(\int_{r_k}^{r_{b_n}} \frac{1}{\left(1+sr^{-\alpha}\right)^L}\frac{2r}{r_{b_n}^2-r_k^2}dr\right)^{b_n-k-1} \nonumber\\
&~~~\times\frac{\exp\left(-2\pi\lambda_b\int_{r_{b_n}}^{\infty}\frac{sr^{-\alpha}}{1+sr^{-\alpha}}rdr\right)}{\left(1+sr_{b_n}^{-\alpha}\right)^L}
\end{align}
for $k=1,2,\ldots,b_n-1$ and
\begin{align}
\mathcal{L}_{I_{r3}}(s)=\exp\left(-2\pi\lambda_b\int_{r_{b_n}}^{\infty}\frac{sr^{-\alpha}}{1+sr^{-\alpha}}rdr\right),
\end{align}
for $k=b_n$, respectively.
\end{lemma}

The proof of Lemma 4 is similar to Appendix A so we omit it here. Since the tractable expression of the coverage probability is complex, we provide more compact forms to bound the coverage probability in the following theorem.

%

\begin{theorem}[{Bound of Coverage Probability of NO-MF}]
The coverage probability of the typical user at SBS $\textbf{d}_k$ out of the $b_n$ serving SBSs with NO-MF transmission, for $k=1,2,\ldots,b_n$, is bounded as:
\begin{align}
P_{\text{cov,no-mf}}^{k,\text{l}}(b_n,\gamma) \leq P_{\text{cov,no-mf}}^{k}(b_n,\gamma) \leq P_{\text{cov,no-mf}}^{k,\text{u}}(b_n,\gamma),
\end{align}
with
\begin{align}
&P_{\text{cov,no-mf}}^{k,\text{u}}(b_n,\gamma)\!\!=\!\!\sum_{l=1}^{L}\!\!\binom{L}{l}(-1)^{l+1}\mathbb{E}_{\delta_k'}\!\!\left[\frac{\beta_4(\delta_k',\eta,\gamma,\alpha,l)}{(1+\beta_2(\eta\delta_k'^{\alpha},\gamma,\alpha,l))^{b_n}}\right],\label{eqn:app3} \\
&P_{\text{cov,no-mf}}^{k,\text{l}}(b_n,\gamma)\!\!=\!\!\sum_{l=1}^{L}\!\!\binom{L}{l}(-1)^{l+1}\mathbb{E}_{\delta_k'}\!\!\left[\frac{\beta_4(\delta_k',1,\gamma,\alpha,l)}{(1+\beta_2(\delta_k'^{\alpha},\gamma,\alpha,l))^{b_n}}\right], \label{eqn:app33}
\end{align}
where
\begin{align}
&\beta_4(\delta_k',x,\gamma,\alpha,l)=\frac{1}{(1+x \gamma l \delta_k'^{\alpha})^L}\times \nonumber\\
&\bigg[\int_{\frac{1}{1+x \gamma l}}^{\frac{1}{1+x \gamma l\delta_k'^{\alpha}}}\frac{2(x \gamma l)^{\frac{2}{\alpha}}}{\alpha(\frac{1}{\delta_k'^2}-1)}
\times v^{\frac{2}{\alpha}-1+L}(1-v)^{-\frac{2}{\alpha}-1}dv\bigg]^{b_n-k-1}
\end{align}
for $k=1,2,\ldots,b_n-1$, where $\delta_k'=\frac{r_k}{r_{b_n}}$ and its pdf can be obtained in a similar way to the pdf of $\delta_k$ in (\ref{eqn:pdf}), and $\beta_4(\delta_k',x,\gamma,\alpha,l)=1$
for $k=b_n$, respectively.
\end{theorem}

The proof of Theorem 4 is similar to Appendix B {and Appendix C}, so we omit it here.
\subsubsection{O-ZF}
Since the SIR distribution of the typical user $u_0$ served by the $k$-th nearest SBS with O-ZF transmission for coded caching is the same as that for probabilistic caching, we can conclude that the coverage probability and approximate coverage probability in O-ZF scheme are given by (\ref{eqn:cov-zf}) and {(\ref{eqn:app22})}, respectively, but for $k=1,2,\ldots,b_n$.

{Similar to the probabilistic caching scenario,we use (approximate) lower bounds as approximate coverage probabilities for the following analysis and optimization for large $L$ and $K$.}
\subsection{Fractional Offloaded Traffic}
In the NO-MF scheme, when $u_0$ requests file $f_n$, the $b_n$ nearest SBSs transmit the cached coded packets of $f_n$ with size $\frac{1}{b_n}$ for each concurrently, and the user decodes these signals successively using SIC. Since the user adopts SIC to decode the signals successively, the signal from $\textbf{d}_k$ can be decoded successfully if and only if $\text{SIR}_k\geq \gamma$ and all the signals from the nearest $k-1$ SBSs have been decoded and subtracted successfully. Thus, the probability that the fraction of $\frac{1}{b_n}$ traffic is successfully offloaded from SBS $\textbf{d}_k$ is given by.
\addtocounter{equation}{1}
\begin{align}
q_k(b_n,\gamma)&=P\left[\underset{i=1,2,\ldots,k}{\bigcap}\text{SIR}_i\geq\gamma\right] \tag{\theequation a}\\
&\simeq \prod_{i=1}^kP_{\text{cov,no-mf}}^{i}(b_n,\gamma), \tag{\theequation b} \label{eqn:a}
\end{align}
where (\ref{eqn:a}) is obtained by assuming the independence of the events $\text{SIR}_i\geq\gamma$, for $i=1,2,\ldots,k$, as in \cite{xu2017modeling} and \cite{wildemeersch2014successive}. Note that the numerical results in \cite{xu2017modeling} show that ignoring the dependency among the events $\text{SIR}_i \geq \gamma$ has negligible impact on the actual AFOT performance.
Therefore, the FOT $L_{n,\text{no-mf}}$ is given by:
\begin{align}
L_{n,\text{no-mf}}=
\begin{cases}
\frac{1}{b_n}\sum_{k=1}^{b_n} q_k(b_n,\gamma),&b_n\in \mathcal{K} \\ \label{eqn:T2}
0,& b_n=\infty
\end{cases},
\end{align}

In the O-ZF scheme, the $b_n$ nearest SBSs transmit the cached coded packets of $f_n$ to $u_0$ sequentially and each coded packet can be successfully decoded if the received SIR exceeds a given SIR target $\gamma$. Hence, the FOT $L_{n,\text{o-zf}}$ is given by:
\begin{align}
L_{n,\text{o-zf}}=
\begin{cases}
\frac{1}{b_n}\sum_{k=1}^{b_n}P_{\text{cov,zf}}^k(K,\gamma),&b_n\in \mathcal{K} \\ \label{eqn:T1}
0,& b_n=\infty
\end{cases}.
\end{align}

\subsection{Ergodic Spectral Efficiency}
In the NO-MF scheme, since the $b_n$ serving SBSs transmit concurrently to $u_0$ at one time slot and they transmit the same amount of information, the actual transmission rate of each SBS is determined by the minimum achievable rate among these $b_n$ SBSs. Thus, the ESE is given by:
\begin{align}
S_{n,\text{no-mf}}&=\mathbb{E}\left[\underset{k=1,2,\ldots,b_n}{\text{min}} \text{log}_2(1+\text{SIR}_{k,\text{no-mf}})\right] \nonumber\\
&=\int_0^\infty P\left[\underset{k=1,2,\ldots,b_n}{\text{min}}\text{log}_2(1+\text{SIR}_{k,\text{no-mf}})>x\right]dx       \nonumber\\
&\overset{(a)}{\simeq}\int_0^\infty \prod_{k=1}^{b_n} P\left[\text{log}_2\left(1+\text{SIR}_{k,\text{no-mf}}\right)>x\right]dx           \nonumber\\
&=\int_0^\infty \prod_{k=1}^{b_n} P_{\text{cov,no-mf}}^{k}(b_n,2^x-1)dx.  \label{eqn:S_mf}
\end{align}
when $b_n\in \mathcal{K}$, where step (a) follows the assumption that the events $\text{SIR}_{k,\text{no-mf}}\geq x$, for $k=1,2,\ldots,b_n$, are independent. When $b_n=\infty$, we have $S_{n,\text{no-mf}}=0$.

In the O-ZF scheme, the typical user needs to connect to the nearest $b_n$ SBSs at different time slots. Therefore, the ESE is defined as the achievable
rate averaged over the $b_n$ serving SBSs and is given by:
\begin{align}
S_{n,\text{o-zf}}=
\begin{cases}
\frac{1}{b_n}\sum_{k=1}^{b_n}R_{k,\text{zf}}(K),&b_n\in \mathcal{K} \\ \label{eqn:S_zf}
0,& b_n=\infty
\end{cases}.
\end{align}


\subsection{Caching Optimization}
In the coded caching model, we want to obtain the optimal cache vector $\textbf{b}$. {Note that we can use approximate coverage probabilities {(\ref{eqn:app22}) or (\ref{eqn:app33})} for cache placement optimization when $K$ and $L$ are large to maximize approximate AFOT or AESE.} The optimization problem can be formulated as:
\addtocounter{equation}{2}
\begin{align}
\text{\textbf{P2:}} \quad \underset{\textbf{b}}{\text{max}} \quad &\sum_{n=1}^N p_nQ_n, \tag{\theequation a}\\
\text{s.t} \quad &\sum_{n=1}^N \frac{1}{b_n} \leq M,      \tag{\theequation b}  \\
 & b_n\in\mathcal{K}\cup\left\{\infty\right\},~~~~n=1,2,\ldots,N, \tag{\theequation c} \label{cache_constraint}
\end{align}
where $Q_n$ can be either the FOT $L_n$ in (\ref{eqn:T2}) and (\ref{eqn:T1}) or ESE $S_n$ in (\ref{eqn:S_mf}) and (\ref{eqn:S_zf}), for both NO-MF and O-ZF schemes. $\textbf{P2}$ is a MCKP which is known to be NP-hard. In the following, we present a property of the optimal cache vector $\textbf{b}^*$, based on which a greedy-based low-complexity algorithm shall be proposed.

{Note that both the exact and approximate coverage probabilities $P_{\text{cov,o-zf}}^{k}(K,\gamma)$ are non-increasing functions of $k$ and $P_{\text{cov,no-mf}}^{k}(b_n,\gamma)$ are non-increasing functions of $b_n$, which can be proved similar to the ZF and MF schemes in Appendix D.} Hence, $L_n$ and $S_n$ in both NO-MF and O-ZF schemes are decreasing functions of $b_n$. Thus, for any two different files $f_i$ and $f_j$ with $i<j$, which means $p_i\geq p_j$, we must have $b_i^*\leq b_j^*$ in order to maximize the AFOT or AESE. Therefore, we have the following theorem.

\begin{theorem}
For any two files $f_i$ and $f_j$ with $1\leq i<j\leq N$, the optimal cache variables must satisfy that $b_i^*\leq b_j^*$.
\end{theorem}

Based on Theorem $5$, we resort to a greedy-based low-complexity algorithm to solve $\textbf{P2}$. We first initialize the file partition value $b_0=1$ and let the initial cache variables $b_n^*=b_0$ for $n=1,2,\ldots,M$ and $b_n^*=\infty$ for $n=M+1,\ldots,N$, which means that the $M$ most popular files are cached entirely in each SBS. Then starting from $b_0=1$, we first identify the last $b_0$ files with $b_n^* = b_0$ as well as the first uncached file with $b_n^*=\infty$, then partition each of these $b_0+1$ files into $b_0+1$ segments by letting their corresponding $b_n^* = b_0+1$. By doing so, a new file is cached without exceeding the total cache size in each SBS. We repeatedly find a set of $b_0+1$ files to update their cache variables until the total profit (AFOT or AESE) cannot be improved further. We then gradually increase the file partition value $b_0$ by one and continue the process until $b_0$ reaches the maximum value of $K$. The details of the algorithm are given in Algorithm \ref{alg:1}. Note that for each file partition value $b_0$, we need to update the cache variables at most $M$ times. For each cache variables update, we need to calculate the total profit and compare it with the previous value, which takes time $\mathcal{O}(N)$. Therefore, the total running time would be $\mathcal{O}(KMN)$.

\begin{algorithm}
\caption{A Greedy-based Low-complexity Algorithm} \label{alg:1}
\begin{algorithmic}[1]
\STATE Initialize the cache vector $\textbf{b}^*=[\underset{M}{\underbrace{1,1,\ldots,1}},\underset{N-M}{\underbrace{\infty,\ldots,\infty}}]$ and $b_0\leftarrow1$;
\WHILE {$b_0<K$}
\STATE Change the cache value of the last $b_0$ files whose $b_n^*=b_0$ as well as the first uncached file in $\textbf{b}^*$ to $b_0+1$ and set this new cache vector as $\textbf{b}$;
\IF {The total profit with cache strategy $\textbf{b}$ is larger than that with $\textbf{b}^*$}
\STATE $\textbf{b}^*\leftarrow\textbf{b}$;
\ELSE
\STATE $b_0\leftarrow b_0+1$;
\ENDIF
\ENDWHILE
\end{algorithmic}
\end{algorithm}

{
\section{Extension to Quantized CSI}
In this section, we analyze the coverage probabilities of the considered transmission schemes when the CSI is imperfect. We focus on the analysis of probabilistic caching as the analysis of coded caching is similar. To model the imperfect CSI, we consider the case where SBSs obtain quantized CSI via limited feedback as in \cite{4641946}. With limited feedback, the channel direction information (CDI) is fed back using a quantization codebook of size $2^B$, where $B$ is the number of feedback bits for each channel. CDI is utilized to design the beamforming vectors.}

{
For MF beamforming, the SIR of the typical user $u_0$ served by $k$-th nearest SBS with quantized CSI is given by:
\begin{equation}
\hat{\text{SIR}}_{k,\text{mf}}=\frac{\hat{g}_{k,\text{mf}}\cdot r_k^{-\alpha}}{\sum_{\textbf{d}_j\in\Phi_b \backslash
\{\textbf{d}_k\}} \hat{g}_{j,\text{mf}}\cdot r_j^{-\alpha}},
\end{equation}
where $\hat{g}_{k,\text{mf}}\sim \Gamma(L,\zeta)$ is the effective channel gain of the desired signal, where $\zeta\triangleq1-2^B\beta(2^B,\frac{L}{L-1})$, $\beta(x,y)=\frac{\Gamma(x)\Gamma(y)}{\Gamma(x+y)}$ is the Beta function, and $\hat{g}_{j,\text{mf}}\sim \exp(1)$ is the effective channel gain of the undesired signal from $\textbf{d}_j$ \cite{4641946}.}
{
Therefore, the coverage probability is given by:
\begin{align}
\!\!\!\!\!\!\hat{P}_{\text{cov,mf}}^{k}(K,\gamma)\!=\!\mathbb{E}_{r_k}\!\!\left[\sum_{i=0}^{L-1}\frac{(-\gamma {r_k}^\alpha/\zeta)^i}{i!}\mathcal{L}_{\hat{I}_{r1}}^{(i)}(\gamma {r_k}^\alpha/\zeta)|r_k\right],
\end{align}
where $\hat{I}_{r1}=\sum_{\textbf{d}_j\in\Phi_b \backslash\{\textbf{d}_k\}} \hat{g}_{j,\text{mf}}\cdot r_j^{-\alpha}=I_{r_1}$.}
{
From the Alzer¡¯s inequality \cite{alzer1997some}, we have $[1-\exp(-\eta x)]^L\leq \int_{0}^x\frac{t^{L-1}\exp(t)}{(L-1)!}dt\leq [1-\exp(-x)]^L$. Therefore, the CDF of $\hat{g}_{k,\text{mf}}$ is bounded as $[1-\exp(-\eta x/\zeta)]^L\leq P\left[\hat{g}_{k,\text{mf}} \leq x\right]\leq [1-\exp(-x/\zeta)]^L$.
Thus, the coverage probability can be upper bounded as:
\begin{align}
\hat{P}_{\text{cov,mf}}^{k}(K,\gamma)\leq\sum_{l=1}^{L}\binom{L}{l}(-1)^{l+1} \mathbb{E}_{r_k}\left[\mathcal{L}_{\hat{I}_{r1}}(\eta \gamma r_k^{\alpha}l/\zeta)|r_k\right].
\end{align}
}
Due to the fact that $\hat{I}_{r1}=I_{r_1}$, the coverage probability of the typical user $u_0$ served by the $k$-th nearest SBS with MF beamforming based on quantized CSI is upper bounded as:
\begin{align}
\hat{P}_{\text{cov,mf}}^{k,\text{u}}(K,\gamma)\!\!=\!\!\sum_{l=1}^{L} \beta_1\left(\eta/\zeta,\gamma,\alpha,l,k\right)\frac{\binom{L}{l}(-1)^{l+1}}{\left(1+\beta_2\left(\eta/\zeta,\gamma,\alpha,l\right)\right)^k},
\end{align}
where the proof is similar to Appendix B. By setting $\eta=1$, the lower bound is also obtained.

In the ZF case, the SIR of $u_0$ when served by SBS $\textbf{d}_k$ with quantized CSI is given by:
\begin{equation}
\hat{\text{SIR}}_{k,\text{zf}}=\frac{\hat{g}_{k,\text{zf}}\cdot r_k^{-\alpha}}{\sum_{\textbf{d}_i\in\Phi_K\backslash \{\textbf{d}_k\}
}\hat{g}_{i,\text{zf}}\cdot r_i^{-\alpha}+\sum_{\textbf{d}_j\in\Phi_b\backslash \Phi_K}\hat{g}_{j,\text{zf}}\cdot r_j^{-\alpha}},
\end{equation}
where $\hat{g}_{k,\text{zf}}\sim \Gamma(L-K+1,\zeta)$ is the effective channel gain of the desired signal, $\hat{g}_{i,\text{zf}}\sim \exp(\frac{1}{1-\zeta})$ is the effective channel gain of the undesired signal for SBSs $\textbf{d}_i\in\Phi_K\backslash \{\textbf{d}_k\}$ and $\hat{g}_{j,\text{zf}}\sim \exp(1)$ is the effective channel gain of the undesired signal for SBSs $\textbf{d}_j\in\Phi_b\backslash \Phi_K$ \cite{7038201}.
{
Similarly, the coverage probability is given by:
\begin{align}
\hat{P}_{\text{cov,zf}}^{k}(K,\gamma)\!\!=\!\!\mathbb{E}_{r_k,r_K}\!\!\left[\sum_{i=0}^{L-K}\frac{(-\gamma {r_k}^\alpha/\zeta)^i}{i!}\mathcal{L}_{\hat{I}_{r2}}^{(i)}(\gamma {r_k}^\alpha/\zeta)|r_k,r_K\right],
\end{align}
where $\hat{I}_{r2}=\sum_{\textbf{d}_i\in\Phi_K\backslash \{\textbf{d}_k\}}\hat{g}_{i,\text{zf}}\cdot r_i^{-\alpha}+\sum_{\textbf{d}_j\in\Phi_b\backslash \Phi_K}\hat{g}_{j,\text{zf}}\cdot r_j^{-\alpha}$. The interference $\hat{I}_{r2} $ consists of two parts, the interference $\hat{I}_1$ from the $k-1$ SBSs $\textbf{d}_i\in\Phi_K\backslash \{\textbf{d}_k\}$ and the interference $\hat{I}_2$ from the SBSs farther than $\textbf{d}_K$. The Laplace transform of $\hat{I}_1$ is given by:
\begin{align}
\mathcal{L}_{\hat{I}_1}(s)&\!\!=\!\!\left(\int_0^{r_k}\!\!\!\! \frac{1}{1+(1-\zeta)sr^{-\alpha}}\frac{2r}{r_k^2}dr\right)^{k-1} \!\!\!\!\!\frac{1}{1+(1-\zeta)sr_K^{-\alpha}} \nonumber\\
&\times\left(\int_{r_k}^{r_K} \frac{1}{1+(1-\zeta)sr^{-\alpha}}\frac{2r}{r_K^2-r_k^2}dr\right)^{K-k-1}
\end{align}
when $k<K$. When $k=K$, we have
\begin{align}
\mathcal{L}_{\hat{I}_1}(s)&=\bigg[1-\frac{2 (1-\zeta)^{2/\alpha}s^{2/\alpha}}{\alpha r_k^2 } \nonumber\\
&\times B\left(\frac{2}{\alpha},1-\frac{2}{\alpha},\frac{1}{1+(1-\zeta)sr_k^{-\alpha}}\right)\bigg]^{K-1}
\end{align}}

{
For the interference $\hat{I}_2$, its Laplace transform is $\mathcal{L}_{\hat{I}_{2}}(s)=\exp\left[-2\pi\lambda_b\int_{r_K}^{\infty}\frac{sr^{-\alpha}}{1+sr^{-\alpha}}rdr\right]$.}
{
Thus, the the coverage probability with ZF beamforming based on quantized CSI is upper bounded as:
\begin{align}
\hat{P}_{\text{cov,zf}}^{k,u}(K,\gamma)&=\sum_{l=1}^{L-K+1}\binom{L-K+1}{l}(-1)^{l+1}   \nonumber\\
&\times \mathbb{E}_{r_k,r_K}\left[\mathcal{L}_{\hat{I}_{r2}}(\kappa \gamma r_k^{\alpha}l/\zeta)|r_k,r_K\right]. \label{eqn:imperfect3}
\end{align}}
{
To simplify the expectation in (\ref{eqn:imperfect3}), we introduce a parameter $\delta_k=\frac{r_k}{r_K}$ and we have
\begin{align}
&\mathbb{E}_{r_k,r_K}\!\!\left[\mathcal{L}_{\hat{I}_{r2}}(\kappa \gamma r_k^{\alpha}l/\zeta)\right]\!\!=\!\!\left[1-\frac{2 \tau_l^{2/\alpha}}{\alpha}B\!\!\left(\frac{2}{\alpha},1-\frac{2}{\alpha},\frac{1}{1+\tau_l}\right)\!\right]^{k-1}  \nonumber\\
&~~~\times\mathbb{E}_{\delta_k}\bigg[\bigg[\frac{2 \tau_l^{2/\alpha}}{\alpha (\delta_k^{-2}-1)}
\bigg(B\left(\frac{2}{\alpha},1-\frac{2}{\alpha},\frac{1}{1+\tau_l}\right)- \nonumber\\
&~~~B\bigg(\frac{2}{\alpha},1-\frac{2}{\alpha},\frac{1}{1+\tau_l\delta_k^{\alpha}}\bigg)\bigg)\bigg]^{K-k-1}\!\!\!\frac{(1+\tau_l\delta_k^{\alpha})^{-1}}{\left[1+\beta_3(\kappa\gamma\delta_k^{\alpha}l/\zeta,\alpha)\right]^K}\bigg]. \label{eqn:imperfect1}
\end{align}
where $\tau_l=\kappa \gamma l\frac{1-\zeta}{\zeta}$ when $k<K$ and when $k=K$, we have
\begin{align}
&\mathbb{E}_{r_k,r_K}\!\!\left[\mathcal{L}_{\hat{I}_{r2}}(\kappa \gamma r_k^{\alpha}l/\zeta)|r_k,r_K\right]\!\!  \nonumber\\
&=\!\!\left[1\!-\!\frac{2\tau_l^{2/\alpha}}{\alpha}B\left(\frac{2}{\alpha},1-\frac{2}{\alpha},\frac{1}{1+\tau_l}\right)\right]^{K-1}\!\!\!\!\times\! \frac{1}{[1+\beta_3(\tau_l 1^{\alpha},\alpha)]^K}. \label{eqn:imperfect2}
\end{align}
The proof is similar to Appendix C. By substituting (\ref{eqn:imperfect1}) or (\ref{eqn:imperfect2}) into (\ref{eqn:imperfect3}), we obtain the upper bound $\hat{P}_{\text{cov,zf}}^{k,u}(K,\gamma)$. By setting $\tau_l=\gamma l\frac{1-\zeta}{\zeta}$, the lower bound is also obtained.}

{With these coverage probabilities, we can formulate and optimize the cache placement problems accordingly. The same algorithm proposed in the previous section can be applied. Similar extension for coded caching with quantized CSI holds and hence is ignored.}

\section{Numerical Results}
In this section, we first validate the tightness of the approximate coverage probability. Then, we demonstrate the performance of probabilistic caching and coded caching by treating MPC as the benchmark. {Finally, we investigate the effects of imperfect CSI in both ZF and MF.} Through these numerical results, we will reveal the role of different beamforming schemes in cache-enabled multi-antenna SCNs.

For presentation convenience, the performances of coded caching obtained by Algorithm \ref{alg:1} with O-ZF and NO-MF scheme are denoted as ``O-ZF-CC'' and ``NO-MF-CC'', respectively. The performances of optimal probabilistic caching using ZF and MF beamforming are denoted as ``ZF-OPC'' and ``MF-OPC'', respectively. For MPC, the performances of using ZF and MF beamforming are denoted as ``ZF-MPC'' and ``MF-MPC'', respectively. \footnote{{Note that we utilize the exact coverage probabilities to optimize the cache strategies in the numerical results for small $K$ ($K=3$). However, when $K$ is large ($K=6$), which causes the exact coverage probabilities hard to obtain due to the high computational complexity, we have to utilize the approximate ones to optimize the cache strategies.}}

The file popularity is assumed to follow the Zipf distribution, i.e., $p_n=\frac{1/n^\delta}{\sum_{j=1}^N1/j^\delta}$ for file $f_n$ with $\delta$ being Zipf skewness parameter. Unless otherwise stated, the other system parameters are set as follows: path loss exponent $\alpha=4$, number of total files $N=100$, cache size $M=10$, Zipf parameter $\delta=0.5$ and cluster size $K=3$.

\subsection{Validation of Analytical Results}
Figs. \ref{fig:MF}$\thicksim$\ref{fig:MF-SIC} compare the analytical results of coverage probabilities in Lemma 1, 2 and 4 with simulation results for all the considered caching and beamforming schemes. {We observe that the simulation and analytical results match well. It is also seen that the (approximate) {lower bounds (\ref{eqn:app11}), (\ref{eqn:app22}) and (\ref{eqn:app33})} are close to the analytical results.} As such, we can use them to approximate the true coverage probabilities for analysis and optimization of AFOT and AESE at large $K$ and $L$.
\begin{figure*}[t]
\vspace{-0.1cm}
\begin{minipage}[t]{0.5\linewidth}
\centering
\includegraphics[width=8.15cm]{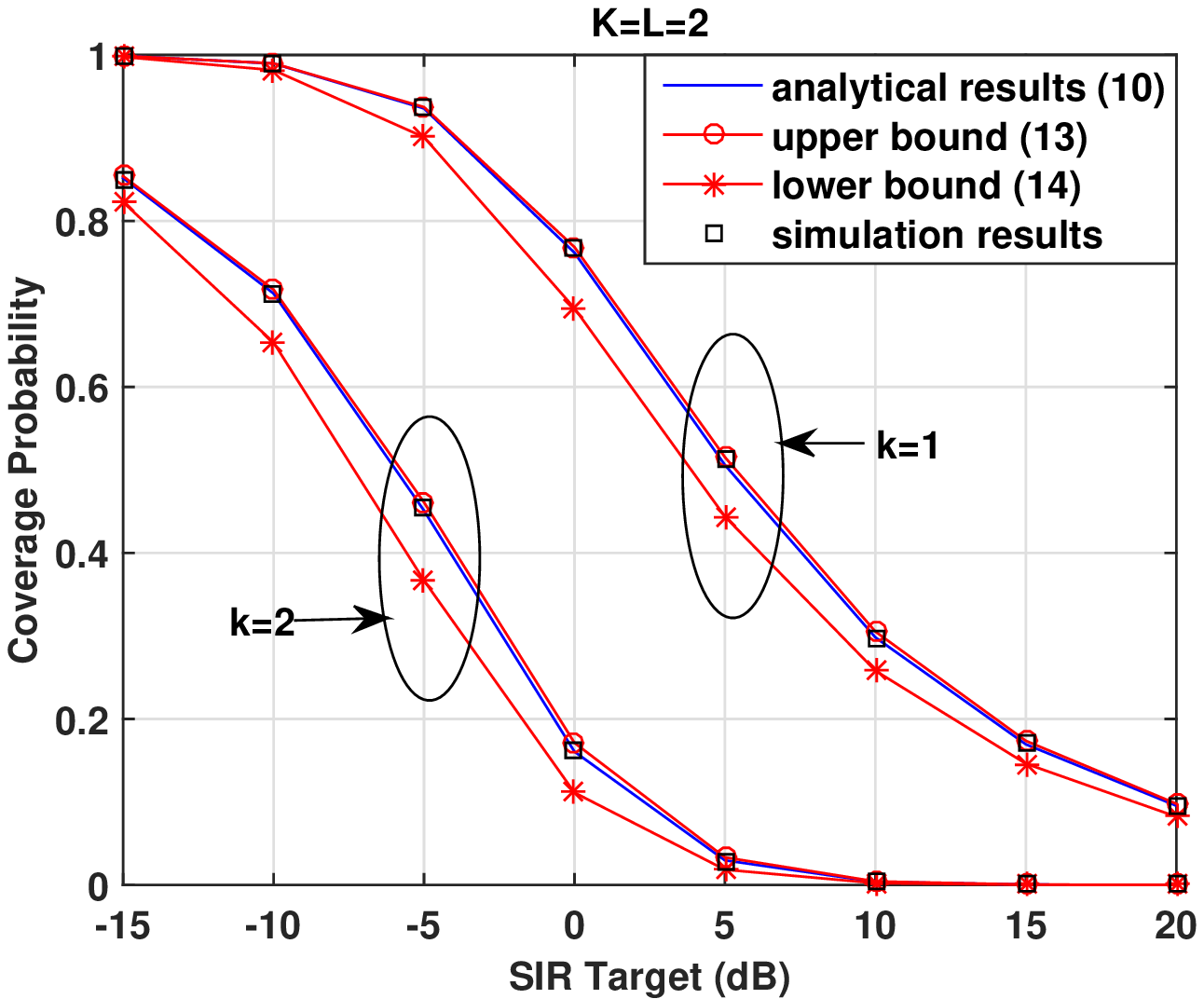}\\
\vspace{-0cm}
\caption{\small{Coverage probability in MF.}}\label{fig:MF}
\end{minipage}
\hfill
\begin{minipage}[t]{0.5\linewidth}
\centering
\includegraphics[width=8.15cm]{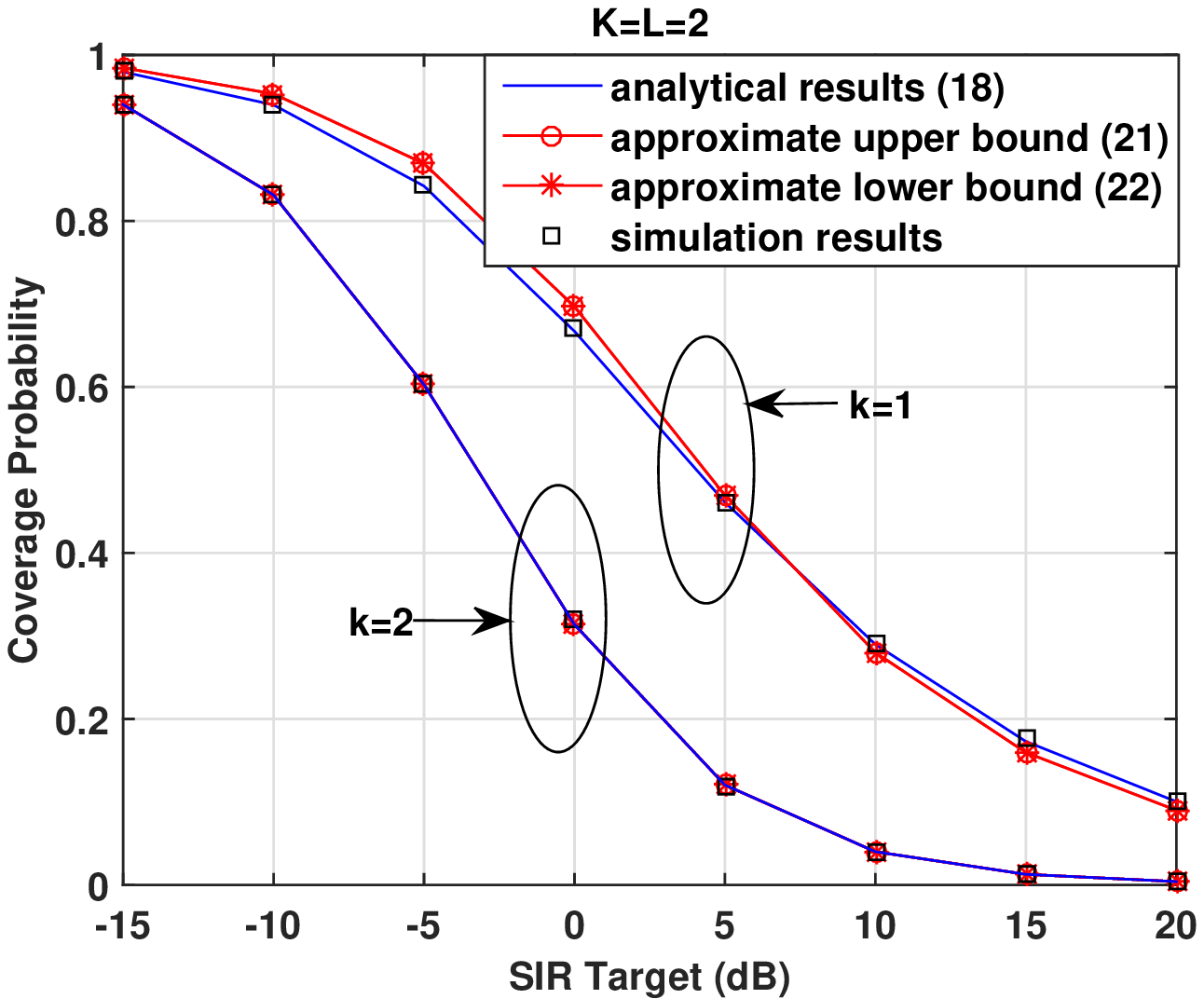}\\
\vspace{-0cm}
\caption{\small{Coverage probability in ZF and O-ZF.}}\label{fig:BF}
\end{minipage}
\vspace{-0.1cm}
\end{figure*}

\begin{figure*}[t]
\vspace{-0.1cm}
\begin{minipage}[t]{0.5\linewidth}
\centering
\includegraphics[width=8.15cm]{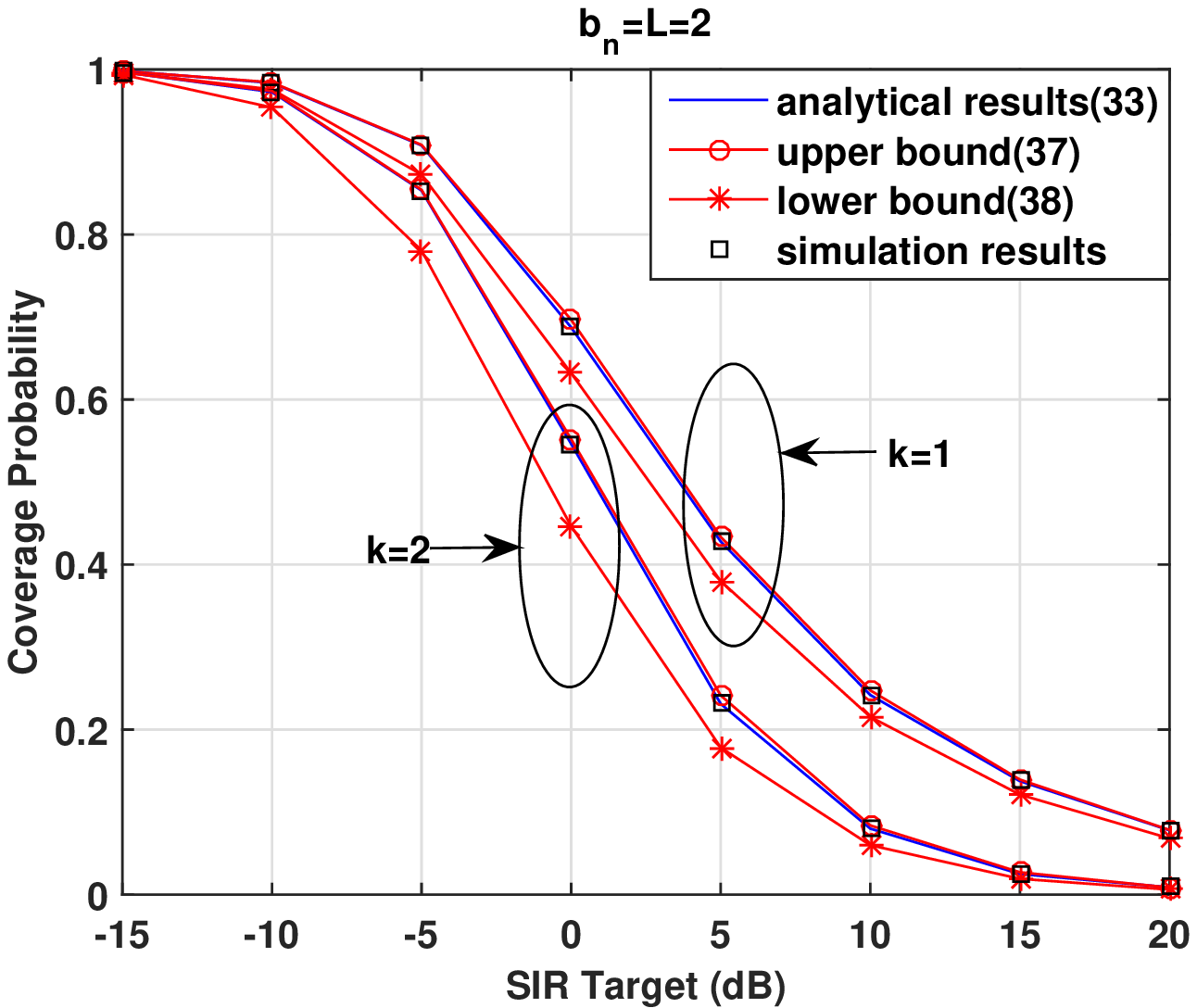}\\
\vspace{-0cm}
 \caption{\small{Coverage probability in NO-MF.}} \label{fig:MF-SIC}
\end{minipage}
\hfill
\begin{minipage}[t]{0.5\linewidth}
\centering
\includegraphics[width=8.15cm]{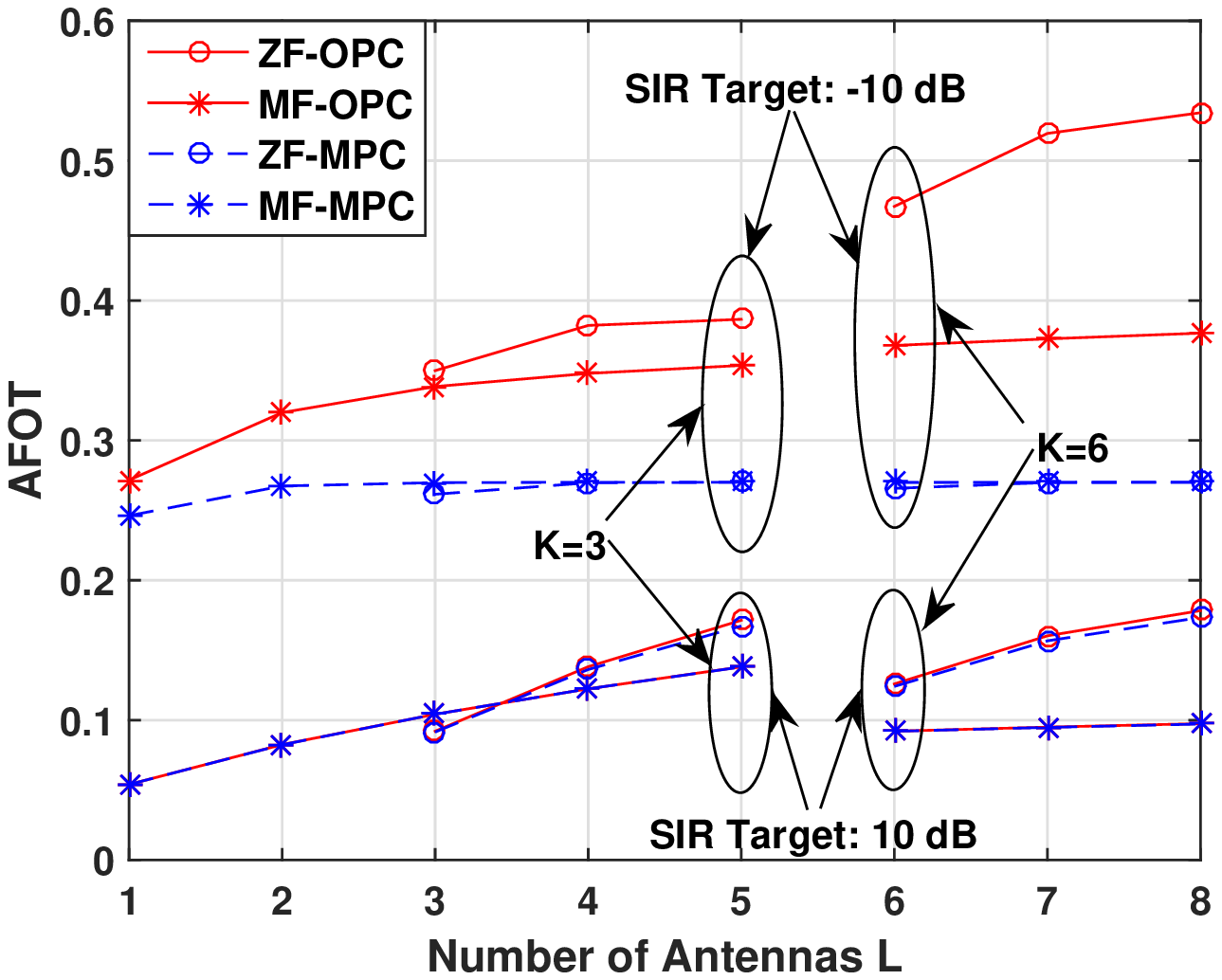}\\
\vspace{-0cm}
\caption{\small{{AFOT for different $L$ in probabilistic caching.}}} \label{fig:opc_1}
\end{minipage}
\vspace{-0.1cm}
\end{figure*}
\subsection{Comparison between MF and ZF for probabilistic caching }
\subsubsection{AFOT}
{Fig. \ref{fig:opc_1}} illustrates the AFOT of OPC and MPC for different number of antennas. It is seen that OPC outperforms MPC at low SIR target ($\gamma=-10$ dB) since it provides chances to users connecting to multiple SBSs rather than the nearest SBS. However, this gain becomes limited at high SIR target ($\gamma=10$ dB). This is because users have a high probability to only connect to the nearest SBS when the decoding threshold is very stringent. Thus, OPC degenerates to MPC at high SIR target. By increasing the number of antennas $L$, AFOT for both ZF and MF are increasing but the gain diminishes as $L$ grows. Besides, the performance gain of OPC over MPC becomes larger when $L$ increases, especially at low SIR target.

Comparing two different beamformings, it is observed that when $L=K$, MF outperforms ZF slightly for OPC at high SIR target and MPC when $K=3$. This is because when the number of antennas equals the cluster size, the effective channel gain of the desired signal with ZF beamforming is much smaller than that of MF beamforming although the former suffers less interference. However, when the number of antennas is larger than the cluster size, SBSs have enough spatial dimensions to null out the intra-cluster interference and strengthen the effective channel gain of the desired signals simultaneously. Therefore, ZF outperforms MF when $L>K$.
\subsubsection{AESE}
\begin{figure*}[t]
\vspace{-0.1cm}
\begin{minipage}[t]{0.5\linewidth}
\centering
\includegraphics[width=8.15cm]{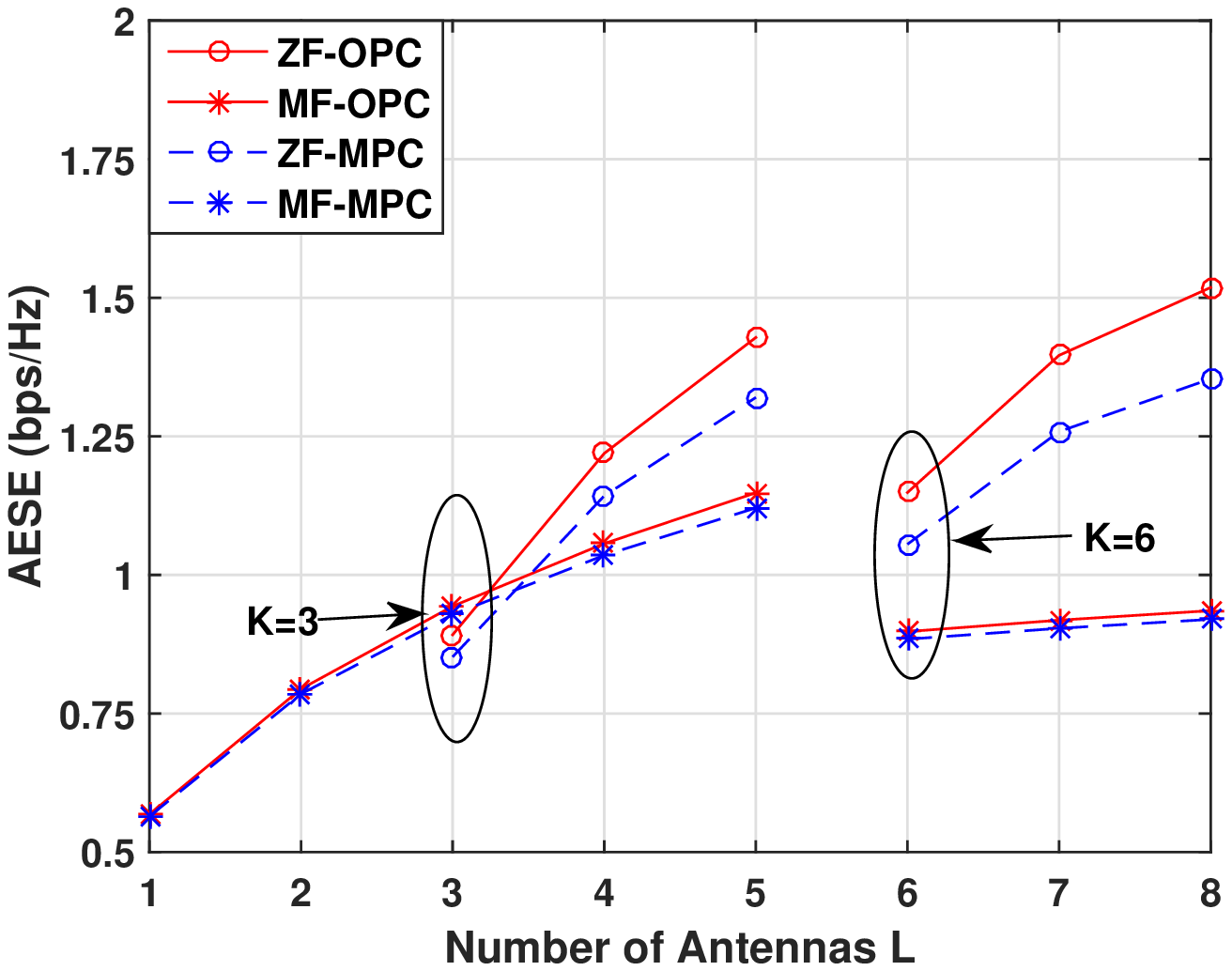}\\
\vspace{-0cm}
\caption{\small{{AESE for different $L$ in probabilistic caching.}}} \label{fig:AESE_2}
\end{minipage}
\hfill
\begin{minipage}[t]{0.5\linewidth}
\centering
\includegraphics[width=8.15cm]{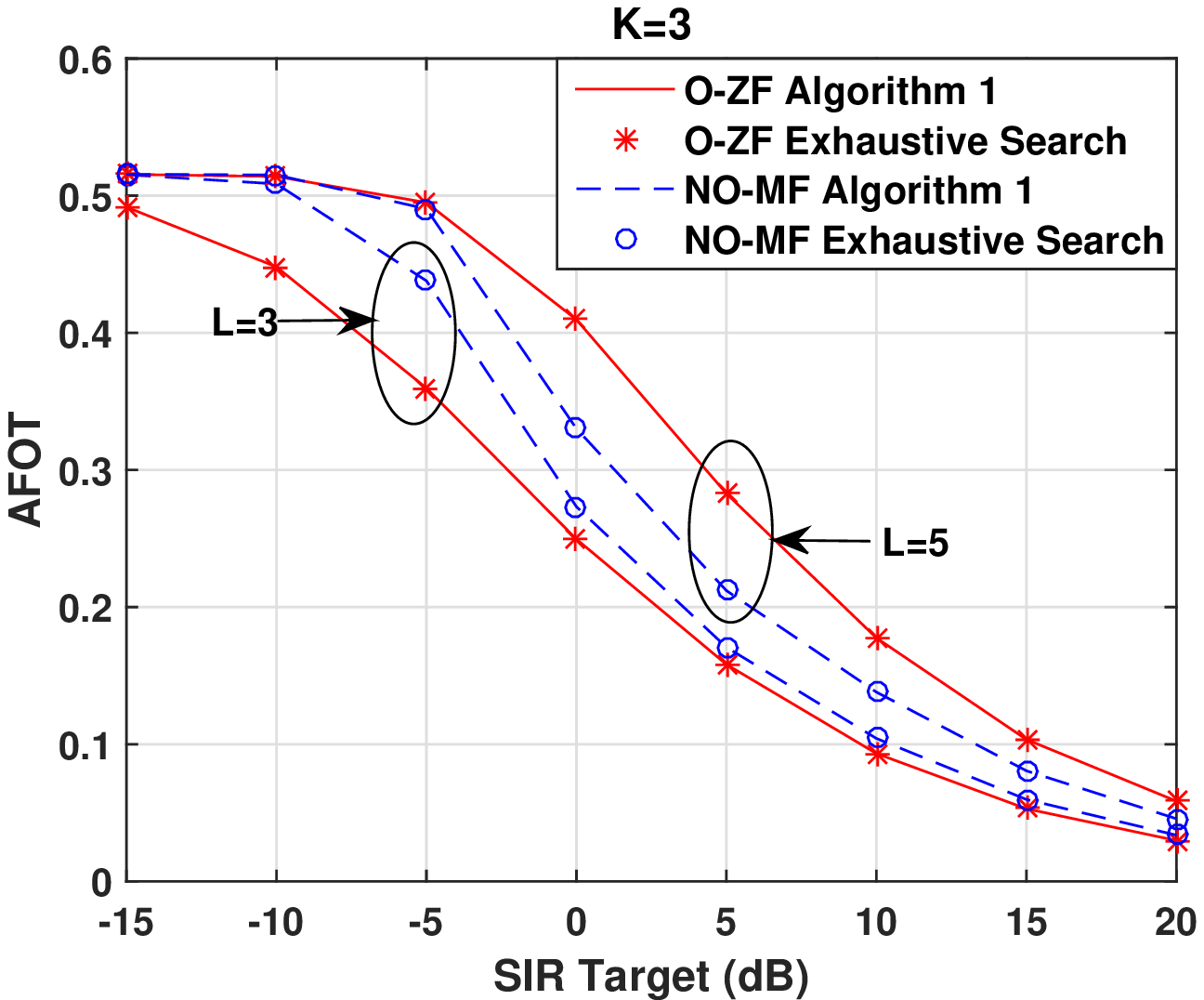}\\
\vspace{-0cm}
\caption{\small{Validation of Algorithm $1$ respect to AFOT.}} \label{fig:Alg1}
\end{minipage}
\vspace{-0.1cm}
\end{figure*}
{Fig. \ref{fig:AESE_2}} illustrates the AESE of OPC and MPC for different number of antennas. It is observed that OPC outperforms MPC in both MF and ZF while the performance gap between OPC and MPC when ZF is applied is larger than that when MF is applied. This is because users suffer strong interference in MF, which causes that the AESE is mainly limited by the performance of the nearest SBS. Besides, the performance gain of OPC over MPC becomes larger by increasing the number of antennas. As such, we can conclude that OPC benefits better from multiple antennas than MPC.

For different beamforming types, {Fig. \ref{fig:AESE_2}} shows that ZF outperforms MF when $L>K$. The reasons are similar to the AFOT case.
\subsection{Comparison between MF and ZF for coded caching }
Since Algorithm \ref{alg:1} is a greedy algorithm, we first validate its effectiveness. For illustration purpose, we only consider the AFOT maximization problem for validation. Fig. \ref{fig:Alg1} shows that the results obtained by Algorithm \ref{alg:1} are almost identical to the optimal solution obtained by exhaustive search. Therefore, we can utilize Algorithm \ref{alg:1} to obtain the coded caching solution.
\subsubsection{AFOT}
Fig. \ref{fig:cc_2} illustrates the AFOT of CC and MPC for different antenna number. At low SIR target ($\gamma=-10$ dB), CC has a great performance gain over MPC since it can make a better utilization of collaborative caching among multiple SBSs. However, at high SIR target ($\gamma=10$ dB), CC outperforms MPC slightly with ZF and it even performs the same as MPC with MF. This is because users have a large probability to only connect to the nearest SBS when the decoding threshold is very stringent, and hence the performance of MPC is close to that of CC. Besides, AFOT for both O-ZF and NO-MF increase as $L$ grows, but the gain is diminishing. Moreover, the performance gap between CC and MPC also becomes larger by increasing $L$, which means CC can enjoy a higher performance gain from multiple antennas than MPC.

{Fig. \ref{fig:cc_2} shows that NO-MF outperforms O-ZF at low SIR target. This is because in NO-MF, the effective channel gain of the desired signal is larger than O-ZF and the strong interference from closer SBSs is canceled simultaneously via SIC-based receiver. When the SIR target is high, users have a large probability to only be served by the nearest SBS. Therefore, O-ZF performs better than NO-MF when $L>K$, which is similar to the probabilistic caching case.}

\subsubsection{AESE}
Fig. \ref{fig:AESE_CC} illustrates the AESE of CC and MPC for different number of antennas. It is observed that when MF beamforming is applied, CC performs almost the same as MPC for all $L$'s. This is because the AESE is limited by the minimum delivery rate of the serving SBSs in the concurrent transmission, and thus users prefer only connecting to the nearest SBS. When ZF beamforming is applied, on the other hand, CC outperforms MPC since it can make a better utilization of multiple SBSs. Moreover, ZF outperforms MF in coded caching. As such, in contrast to the previous finding in \cite{xu2017modeling} where CC performs nearly to MPC in terms of AESE in the single-antenna system, we find that CC outperforms considerably MPC in the multi-antenna case if proper beamforming is chosen.


\begin{figure*}[t]
\vspace{-0.1cm}
\begin{minipage}[t]{0.5\linewidth}
\centering
\includegraphics[width=8.15cm]{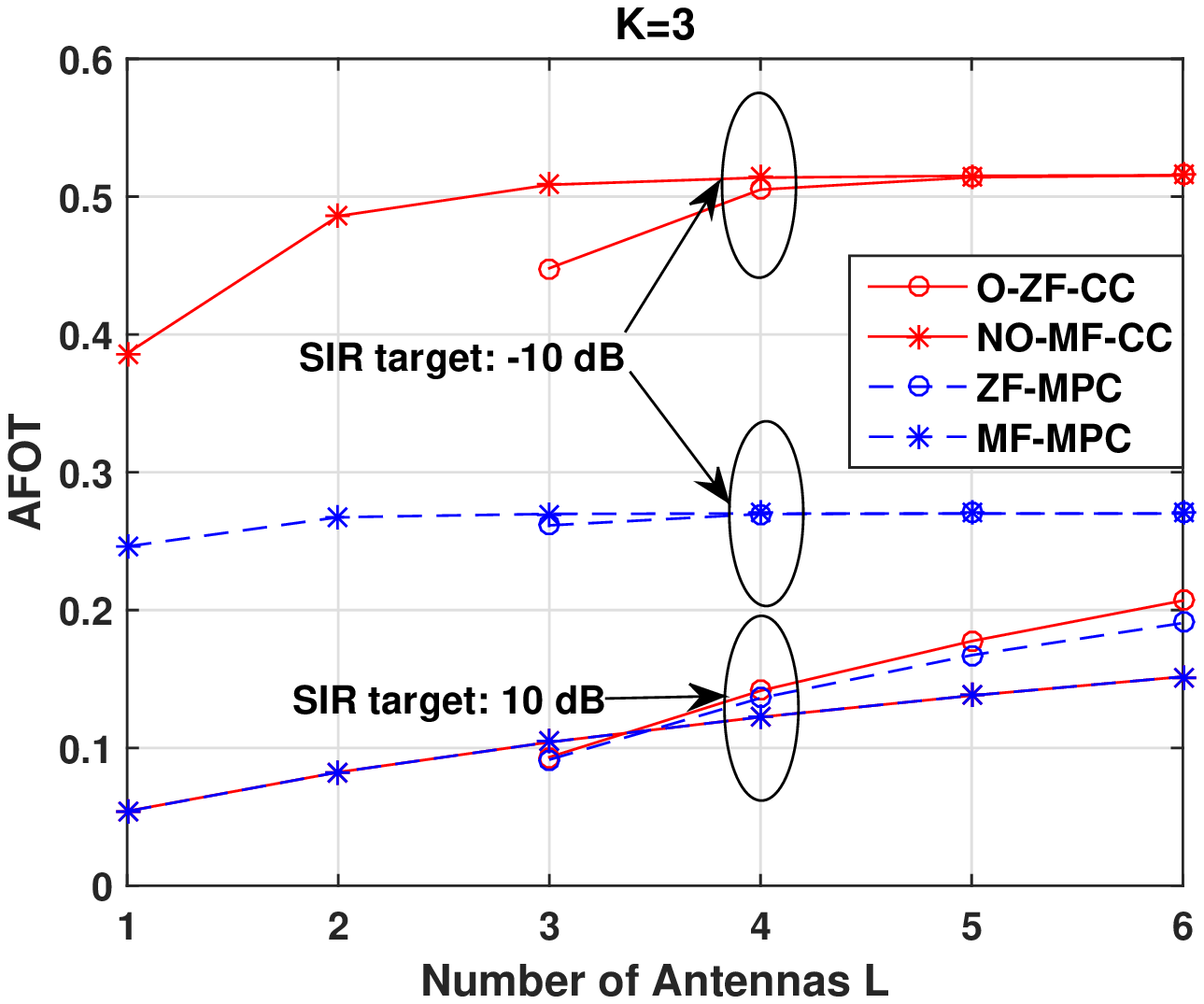}\\
\vspace{-0cm}
\caption{\small{AFOT for different $L$ in coded caching.}} \label{fig:cc_2}
\end{minipage}
\hfill
\begin{minipage}[t]{0.5\linewidth}
\centering
\includegraphics[width=8.15cm]{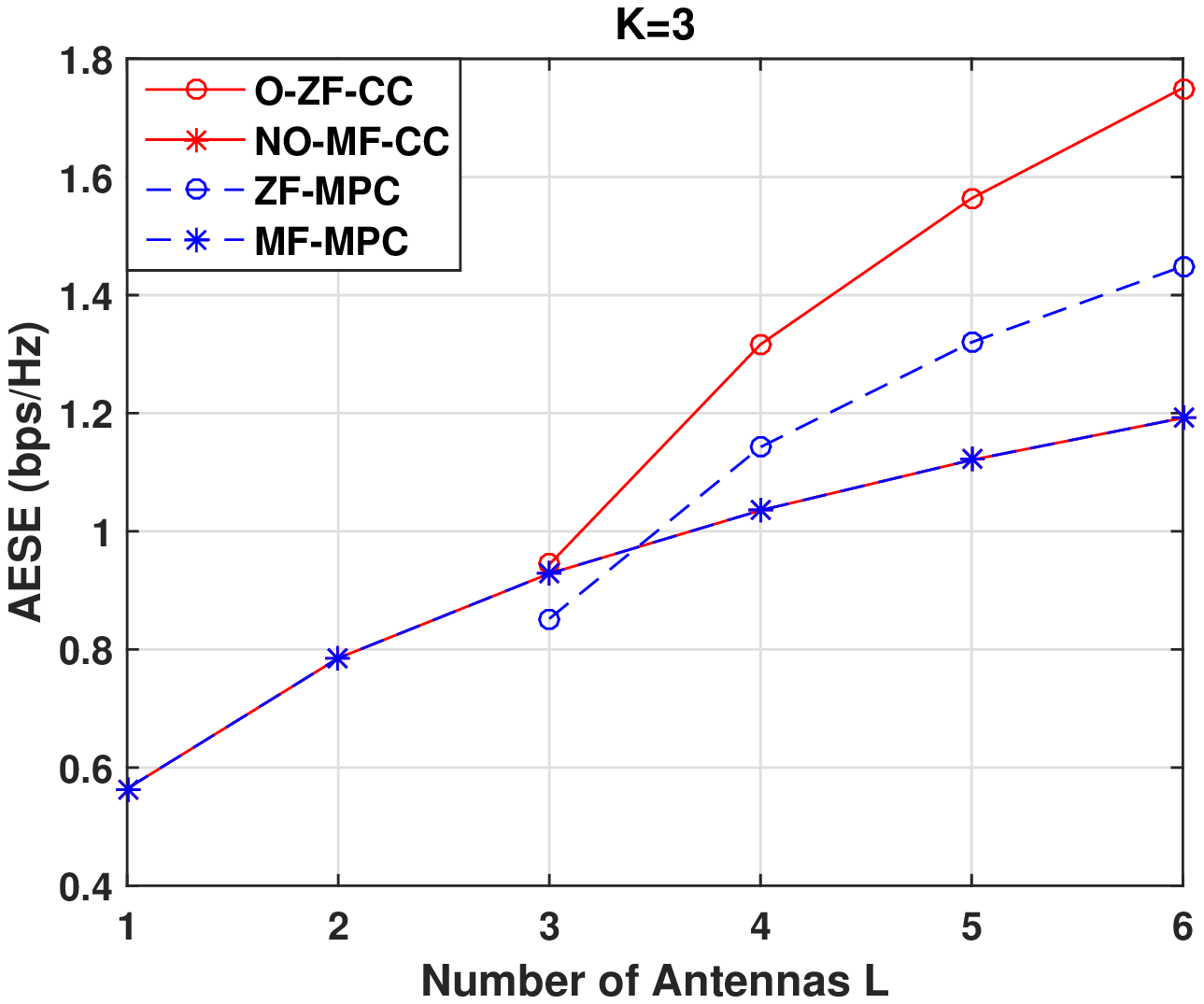}\\
\vspace{-0cm}
\caption{\small{AESE for different $L$ in coded caching.}} \label{fig:AESE_CC}
\end{minipage}
\vspace{-0.1cm}
\end{figure*}

\subsection{Impact of Imperfect CSI}
\begin{figure*}[t]
\vspace{-0.1cm}
\begin{minipage}[t]{0.5\linewidth}
\centering
\includegraphics[width=8.15cm]{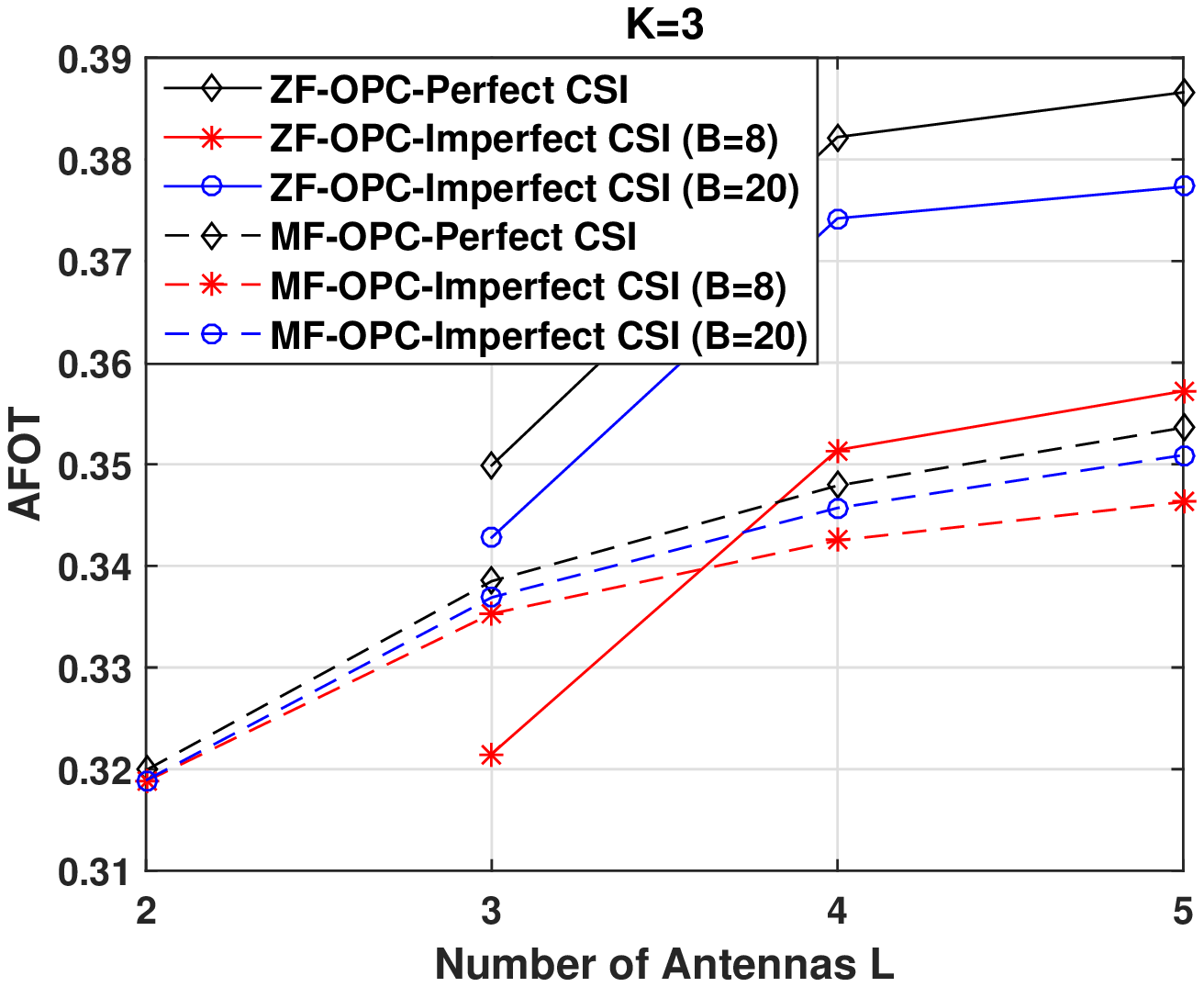}\\
\vspace{-0cm}
\caption{\small{AFOT for different $L$ in probabilistic caching with imperfect CSI ($\gamma=-10$ dB).}} \label{fig:CSI1}
\end{minipage}
\hfill
\begin{minipage}[t]{0.5\linewidth}
\centering
\includegraphics[width=8.15cm]{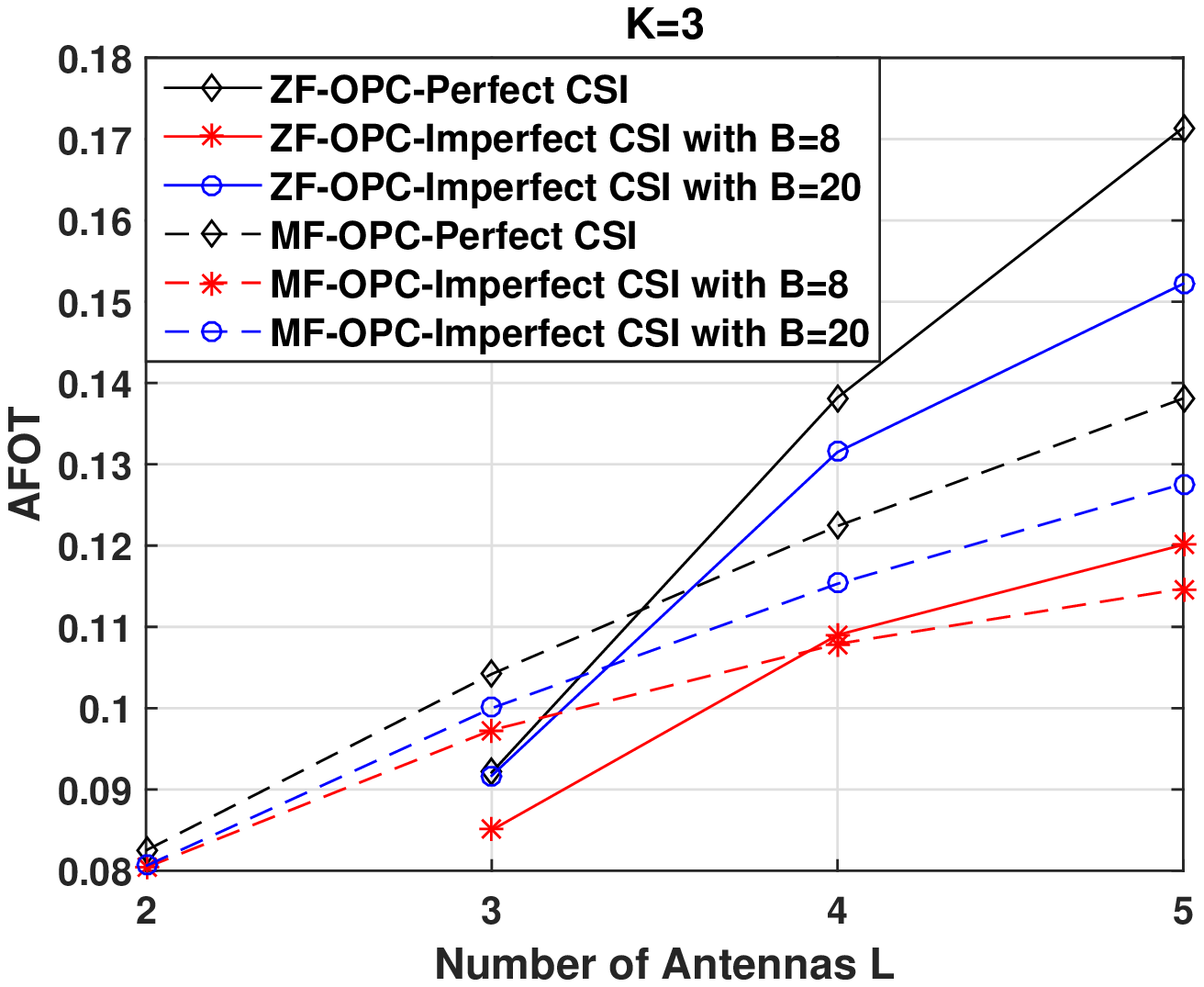}\\
\vspace{-0cm}
\caption{\small{AFOT for different $L$ in probabilistic caching with imperfect CSI ($\gamma=10$ dB).}} \label{fig:CSI2}
\end{minipage}
\vspace{-0.1cm}
\end{figure*}
Fig. \ref{fig:CSI1} and Fig. \ref{fig:CSI2} show that imperfect CSI degrades the performance of caching in both ZF and MF. Comparing ZF and MF, it is observed MF is more robust than ZF with imperfect CSI since each SBS needs to know more CSI to do the coordination when ZF is applied. Besides, the performance comparison between ZF and MF for differen $L$ with imperfect CSI is similar to the perfect CSI case. Moreover, when the number of feedback bits $B$ increases, the performance of imperfect CSI gets closer to that of perfect CSI case.

\section{Conclusion}
In this work, we analyze and optimize the probabilistic caching and coded caching in cache-enabled multi-antenna SCNs. We propose a user-centric SBS clustering and transmission framework, which allows each user to connect with the $K$ nearest SBSs within its cluster. We obtain approximate and compact integral expressions of AFOT and AESE, respectively, with MF and ZF beamforming. Then we formulate the cache placement problem to find the optimal cache solutions. The probabilistic cache placement problem is shown to be convex and optimal solutions are obtained. The coded cache placement problem is a MCKP and we solve it with a greedy-based low-complexity algorithm efficiently. {We also extend the analysis and optimization above to the imperfect case.} Numerical results show that multiple antennas can boost the advantage of probabilistic caching and coded caching over the traditional most popular caching with the proper use of beamforming. Numerical results also demonstrate the performance difference between MF and ZF under different number of antennas {in both perfect and imperfect CSI cases.}

\begin{appendices}
\begin{center}
Appendix A: Proof of Lemma 1
\end{center}

The coverage probability $P_{\text{cov,mf}}^{k}(K,\gamma)$ can be written as:
\begin{align}
&P_{\text{cov,mf}}^{k}(K,\gamma)=\mathbb{E}_{r_k,I_{r1}}\left[P\left[g_{k,\text{mf}} \geq \gamma r_k^{\alpha} I_{r1}\right]|r_k,I_{r1}\right] \nonumber \\
&~~~~\overset{(a)}{=}\mathbb{E}_{r_k,I_{r1}}\left[\sum_{i=0}^{L-1}\frac{(\gamma r_k^{\alpha} I_{r1})^i}{i!}e^{-\gamma r_k^{\alpha} I_{r1}}|r_k,I_{r1}\right]   \nonumber \\
&~~~~\overset{(b)}{=}\mathbb{E}_{r_k}\left[\sum_{i=0}^{L-1}\frac{(-\gamma {r_k}^\alpha)^i}{i!}\mathcal{L}_{I_{r1}}^{(i)}(\gamma {r_k}^\alpha)|r_k\right],
\end{align}
where (a) follows from the series expansion of the CCDF $\overline{F}(x;m,\theta)$ for Gamma distribution $\Gamma(m,\theta)$ when $\theta$ is a positive integer, i.e., $\overline{F}(x;m,\theta)=\sum_{i=0}^{m-1}\frac{1}{i!}(\frac{x}{\theta})^i e^{-\frac{x}{\theta}}$; and (b) follows from the derivative property of the Laplace transform: $\mathbb{E}[X^ie^{-sX}]=(-1)^i \mathcal{L}_X^{(i)}(s)$.

The interference $I_{r1}$ consists of two parts, the interference $I_1$ from the $k-1$ SBSs closer to $u_0$ than the serving SBS $\textbf{d}_k$ and the interference $I_2$ from the SBSs farther than $\textbf{d}_k$.
The Laplace transform of $I_1$ is given by:
\begin{align}
\mathcal{L}_{I_1}(s)&=\mathbb{E}_{\Phi_b,g_{j,\text{mf}}}\left[\prod_{\textbf{d}_j\in\Phi_b\bigcap \mathcal{B}(0,r_k) \backslash \{\textbf{d}_k\}}\exp\left(-sg_{j,\text{mf}}\cdot r_j^{-\alpha}\right)\right] \nonumber \\
&\overset{(a)}{=}\mathbb{E}_{\Phi_b}\left[\prod_{\textbf{d}_j\in\Phi_b\bigcap \mathcal{B}(0,r_k) \backslash \{\textbf{d}_k\}}\frac{1}{1+sr_j^{-\alpha}}\right]  \nonumber \\
&\overset{(b)}{=}\left(\int_0^{r_k} \frac{1}{1+sr^{-\alpha}}\frac{2r}{r_k^2}dr\right)^{k-1}, \label{eqn:L1}
\end{align}
where (a) follows from  i.i.d. exponential distribution with unit mean of $g_{j,\text{mf}}$ and it is also independent with the HPPP $\Phi_b$. Step (b) follows from that the locations of $K-1$ SBSs are independently and uniformly distributed in the circle area $\mathcal{B}(0,r_k)\triangleq \{\textbf{x}\in\mathbb{R}^2| \Vert\textbf{x}\Vert\leq r_k\}$.
The Laplace transform of $I_2$ is given by:
\begin{align}
\mathcal{L}_{I_2}(s)
&=\mathbb{E}_{\Phi_b}\left[\prod_{\textbf{d}_j\in\Phi_b\backslash \mathcal{B}(0,r_k)}\frac{1}{1+sr_j^{-\alpha}}\right] \nonumber\\
&=\exp\left(-2\pi\lambda_b\int_{r_k}^{\infty}\frac{sr^{-\alpha}}{1+sr^{-\alpha}}rdr\right), \label{eqn:L2}
\end{align}
where the last step follows from the probability generating functional (PGFL) of the HPPP. Thus, Lemma 1 is proved.
\begin{center}
Appendix B: Proof of Theorem 1
\end{center}
%

%

From the Alzer's Inequality \cite{huang2007space, alzer1997some}, the the coverage probability is upper bounded as:
\begin{align}
&P_{\text{cov,mf}}^{k}(K,\gamma)=\mathbb{E}_{r_k,I_{r1}}\left[P\left[g_{k,\text{mf}} \geq \gamma r_k^{\alpha} I_{r1}\right]|r_k,I_{r1}\right] \nonumber \\
&~~~\leq \sum_{l=1}^{L}\binom{L}{l}(-1)^{l+1} \mathbb{E}_{r_k,I_{r1}}\left[e^{-\eta \gamma r_k^{\alpha} I_{r1} l}|r_k,I_{r1}\right]  \nonumber \\
&=\sum_{l=1}^{L}\binom{L}{l}(-1)^{l+1} \mathbb{E}_{r_k}\left[\mathcal{L}_{I_{r1}}(\eta \gamma r_k^{\alpha}l)|r_k\right]. \label{eqn:22}
\end{align}
To simplify (\ref{eqn:22}), we first rewrite the Laplace transform of the interference $I_{r1}$ as:
\begin{align}
\mathcal{L}_{I_{r1}}(s)&=\left(\int_0^{r_k} \frac{1}{1+sr^{-\alpha}}\frac{2r}{r_k^2}dr\right)^{k-1} \nonumber \\
& \times\exp\left(-2\pi\lambda_b\int_{r_k}^{\infty}\frac{sr^{-\alpha}}{1+sr^{-\alpha}}rdr\right) \nonumber \\
&=\left[1-\frac{2 s^{2/\alpha}}{\alpha r_k^2 }B\left(\frac{2}{\alpha},1-\frac{2}{\alpha},\frac{1}{1+sr_k^{-\alpha}}\right)\right]^{k-1}  \nonumber \\
&\times \exp \left[-2\pi\lambda_b \frac{s^{\frac{2}{\alpha}}}{\alpha}B^{'}\left(\frac{2}{\alpha},1-\frac{2}{\alpha},\frac{1}{1+sr_k^{-\alpha}}\right)\right],
\end{align}
where the last step follows by first replacing $s^{-\frac{1}{\alpha}}r$ with $u$, then replacing $\frac{1}{1+u^{-\alpha}}$ with $v$. Therefore, the Laplace transform in (\ref{eqn:22}) can be written as:
\begin{align}
\mathcal{L}_{I_{r1}}(\eta \gamma r_k^{\alpha}l)&=\beta_1(\eta,\gamma,\alpha,l,k) \exp\left(-\pi\lambda_b r_k^2 \beta_2\left(\eta,\gamma,\alpha,l\right)\right),
\end{align}
where $\beta_1(\eta,\gamma,\alpha,l,k)$ and $\beta_2(\eta,\gamma,\alpha,l)$ are defined for notation simplicity, given by (\ref{eqn:beta1}) and (\ref{eqn:beta2}), respectively, with $x=\eta$. Next we need to calculate the expectation of $\mathcal{L}_{I_{r1}}(\eta \gamma r_k^{\alpha}l)$ over $r_k$ from (\ref{eqn:22}). It is observed that only the inter-cluster interference is related to $r_k$. Therefore, evaluating the expectation of $\exp(-\pi\lambda_b r_k^2 \beta_2(\eta,\gamma,\alpha,l))$ is enough, which is given by:
\begin{align}
&\mathbb{E}_{r_k}\left[\exp\left(-\pi\lambda_b r_k^2 \beta_2(\eta,\gamma,\alpha,l)\right)\right] \nonumber \\
&=\int_0^\infty \exp(-\pi\lambda_b r_k^2 \beta_2(\eta,\gamma,\alpha,l)) \frac{2(\lambda_b\pi r_k^2)^k}{r_k\Gamma(k)}\exp(-\lambda_b \pi r_k^2) dr_k \nonumber \\
&\overset{(a)}{=}\int_0^\infty  \left[\frac{z}{1+\beta_2(\eta,\gamma,\alpha,l)}\right]^{k-1}\!\!\times\frac{e^{-z}}{\Gamma(k)\left(1+\beta_2(\eta,\gamma,\alpha,l)\right)}dz \nonumber \\
&\overset{(b)}{=}\left[\frac{1}{1+\beta_2(\eta,\gamma,\alpha,l)}\right]^k, \label{eqn:expectation}
\end{align}
where step (a) follows from the change of variable $z=\pi \lambda_b r_k^2 (1+\beta_2(\eta,\gamma,\alpha,l))$
and step (b) follows from the Gamma distribution property that $\int_0^\infty t^k e^{-\lambda t}dt=\frac{k!}{\lambda^{k+1}}$.

By substituting (\ref{eqn:expectation}) into (\ref{eqn:22}), we obtain the upper bound of the coverage probability (\ref{eqn:app1}). The lower bound (\ref{eqn:app11}) can be similarly proved by letting $\eta=1$ in the above derivations. This theorem is thus proved.
\begin{center}
Appendix C: Proof of Theorem 2
\end{center}

Similar to the proof of Theorem $1$, the coverage probability in ZF can be upper bounded by:
\begin{align}
P_{\text{cov,zf}}^{k}(K,\gamma) &\leq \sum_{l=1}^{L-K+1}\binom{L-K+1}{l}(-1)^{l+1}  \nonumber\\
&~~~\times \mathbb{E}_{r_k,r_K}[\mathcal{L}_{I_{r2}}(\kappa \gamma r_k^{\alpha}l)|r_k,r_K]. \label{eqn:upper bound}
\end{align}
{By substituting $r_K$ for $r_k$ in (\ref{eqn:L2}), we obtain the Laplace transform of inter-cluster interference, which is given by:}  
\begin{align}
\mathcal{L}_{I_{r2}}(s)&=\exp\left[-2\pi\lambda_b\int_{r_K}^{\infty}\frac{sr^{-\alpha}}{1+sr^{-\alpha}}rdr\right].  \label{eqn:L_2}
\end{align}
By introducing a parameter $\delta_k=\frac{r_k}{r_K}$, the Laplace transform in (\ref{eqn:upper bound}) can be written as:
\begin{align}
\mathcal{L}_{I_{r2}}(\kappa \gamma r_k^{\alpha}l)&=\exp\left(-2\pi\lambda_b\int_{r_K}^\infty\frac{r^{-\alpha}\kappa \gamma r_k^{\alpha}l}{1+r^{-\alpha}\kappa \gamma r_k^{\alpha}l}rdr\right) \nonumber \\
&=\exp\left(-2\pi\lambda_b\int_{r_K}^\infty\frac{r}{1+(\frac{r}{r_K})^{\alpha}(\kappa \gamma \delta_k^{\alpha}l)^{-1}}dr\right) \nonumber \\
&=\exp\left(-\pi\lambda_b r_K^2(\kappa \gamma \delta_k^{\alpha}l)^{\frac{2}{\alpha}}\int_{(\kappa \gamma \delta_k^{\alpha}l)^{-\frac{2}{\alpha}}}^{\infty}\frac{1}{1+v^{\frac{\alpha}{2}}}dv\right),\label{eqn:Laplace}
\end{align}
where the last step follows from the change of variable $v=\left[\frac{r}{r_K}\left(\frac{1}{\kappa \gamma \delta_k^{\alpha}l}\right)^{\frac{1}{\alpha}}\right]^2$.
For notation simplicity, we let
\begin{align}
\beta_3(\kappa \gamma \delta_k^{\alpha}l,\alpha)=(\kappa \gamma \delta_k^{\alpha}l)^{\frac{2}{\alpha}}\int_{(\kappa \gamma \delta_k^{\alpha}l)^{-\frac{2}{\alpha}}}^{\infty}\frac{1}{1+v^{\frac{\alpha}{2}}}dv.  \label{eqn:beta3}
\end{align}
From (\ref{eqn:Laplace}), it is observed that we need to calculate the expectation over $\delta_k$ and $r_K$, rather than $r_k$ and $r_K$ as in (\ref{eqn:upper bound}). Thus, we first calculate the expectation of (\ref{eqn:Laplace}) over $r_K$.
\begin{align}
&\mathbb{E}_{r_K}\left[\mathcal{L}_{I_{r2}}\left(\kappa \gamma (\delta_k r_K)^{\alpha}l\right)|\delta_k,r_K\right] \nonumber\\
&=\int_0^\infty \exp(-\pi\lambda_b r_K^2 \beta_3(\kappa \gamma \delta_k^{\alpha}l,\alpha))\times \frac{2(\lambda_b\pi r_K^2)^K}{r_K\Gamma(K)} \nonumber\\
&~~~~\times \exp(-\lambda_b \pi r_K^2) dr_K  \nonumber\\
&=\frac{1}{\left[1+\beta_3(\kappa \gamma \delta_k^{\alpha}l,\alpha)\right]^K}, \label{eqn:beta33}
\end{align}
where the last step follows from the change of variables similar to (\ref{eqn:expectation}). Therefore, $P_{\text{cov,zf}}^{k,\text{u}}(K,\gamma)$ is given by:
\begin{align}
&P_{\text{cov,zf}}^{k,\text{u}}(K,\gamma)=\mathbb{E}_{\delta_k}\left[\sum_{l=1}^{L-K+1} \frac{\binom{L-K+1}{l}(-1)^{l+1}}{\left[1+\beta_3(\kappa \gamma \delta_k^{\alpha}l,\alpha)\right]^K}\right]. \label{eqn:upper bound1}
\end{align}

To obtain the expectation above over $\delta_k$, we first need to know the pdf of $\delta_k$. Utilizing the joint pdf of $r_k$ and $r_K$ given in (\ref{eqn:joint pdf}), the CDF of $\delta_k$ is given by:
\begin{align}
P[\delta_k \leq x]&=P[r_k \leq xr_K] \nonumber\\
&=\int_0^\infty \int_0^{xr_K} f_{R_k,R_K}(r_k,r_K)dr_kdr_K \nonumber\\
&=\int_0^\infty \int_0^{xr_K}\frac{4r_k r_k^{2(k-1)}r_K}{\Gamma(K-k)\Gamma(k)}(\lambda_b\pi)^K  \nonumber\\
&~~~\times (r_K^2-r_k^2)^{K-k-1}\exp(-\lambda_b \pi r_K^2)dr_kdr_K \nonumber\\
&=1-\sum_{i=0}^{k-1}\frac{(K-1)!x^{2(k-1-i)}(1-x^2)^{K-k+i}}{(K-k+i)!(k-1-i)!}, \label{eqn:CDF}
\end{align}
where $0\leq x \leq 1$. Then, the pdf of $\delta_k$ can be obtained as:
\begin{align}
f_{\delta_k}(x)&=\frac{dP[\delta_k \leq x]}{dx} \nonumber\\
&=\sum_{i=0}^{k-1} \frac{(K-1)!\left[(K-1)x^2-(k-i-1)\right]}{(K-k+i)!(k-1-i)!} \nonumber\\
&~~~\times 2x^{2(k-1-i)-1}(1-x^2)^{K-k+i-1} \nonumber\\
&=\frac{2(K-1)!}{(k-1)!(K-k-1)!}x^{2k-1}(1-x^2)^{K-k-1}.  \label{eqn:pdf}
\end{align}

Recall (\ref{eqn:beta3}), we approximate the integral in it as a constant value according to randomness of $\delta_k$. {By utilizing partial integration, we can calculate that $\mathbb{E}(\delta_k^2)=\frac{k}{K}$. Therefore, we approximate the integral in (\ref{eqn:beta3}) as a constant value according to randomness of $\delta_k$ similar to \cite[Eqn. (28)]{lee2015spectral} as:
\begin{align}
&\mathbb{E}\left[\int_{\delta_k^{-2}(\kappa \gamma l)^{-\frac{2}{\alpha}}}^{\infty}\frac{1}{1+v^{\frac{\alpha}{2}}}dv\right]=\mathbb{E}\left[\mathcal{A}\left(\frac{(\kappa \gamma l)^{-\frac{2}{\alpha}}}{\delta_k^2}\right)\right]\nonumber\\
&\simeq \sqrt{\mathbb{E}(\delta_k^2)} \mathcal{A}\left(\frac{(\kappa \gamma l)^{-\frac{2}{\alpha}}}{\sqrt{\mathbb{E}(\delta_k^2)}}\right)=\sqrt{\frac{k}{K}} \mathcal{A}\left(\frac{\sqrt K (\kappa \gamma l)^{-\frac{2}{\alpha}}}{\sqrt k }\right).
\end{align}
}


Thus, we can approximate $\beta_3(\kappa \gamma \delta_k^{\alpha}l,\alpha)$ as $\beta_3(\kappa \gamma \delta_k^{\alpha}l,\alpha)\simeq \delta_k^2(\kappa \gamma l)^{\frac{2}{\alpha}}\sqrt{\frac{k}{K}} \mathcal{A}\left(\frac{\sqrt K (\kappa \gamma l)^{-\frac{2}{\alpha}}}{\sqrt k }\right)$.
Therefore, we have
\begin{align}
&\mathbb{E}_{\delta_k}\left[\left(\frac{1}{1+\beta_3(\kappa \gamma \delta_k^{\alpha}l,\alpha)}\right)^K\right] \nonumber \\
&~~~=\int_0^1\left[\frac{1}{1+\beta_3(\kappa \gamma x^{\alpha}l,\alpha)}\right]^K f_{\delta_k}(x)dx \nonumber \\
&~~~\simeq\int_0^1\frac{f_{\delta_k}(x)}{\left[1+(\kappa \gamma l)^{\frac{2}{\alpha}}\sqrt{\frac{k}{K}} \mathcal{A}\left(\frac{\sqrt K (\kappa \gamma l)^{-\frac{2}{\alpha}}}{\sqrt k }\right)x^2\right]^K} dx \nonumber \\ &~~~=\frac{1}{{\left[1+(\kappa \gamma l)^{\frac{2}{\alpha}}\sqrt{\frac{k}{K}} \mathcal{A}\left(\frac{\sqrt K (\kappa \gamma l)^{-\frac{2}{\alpha}}}{\sqrt k }\right)\right]^k}}.  \label{eqn:Edelta}
\end{align}

Substituting (\ref{eqn:Edelta}) into (\ref{eqn:upper bound1}), we obtain the approximate upper bound of the coverage probability (\ref{eqn:app2}). The approximate lower bound (\ref{eqn:app22}) can be similarly proved by letting $\kappa=1$ in the above derivations. This theorem is thus proved.
\begin{center}
Appendix D: Proof of Lemma 3
\end{center}

{To prove the convexity of $\textbf{P1}$, we first prove that the coverage probabilities in both MF and ZF schemes are non-increasing functions of $k$. For the exact coverage probabilities, this property holds obviously. For the lower bound (\ref{eqn:app11}) for MF, we have
\begin{align}
&P_{\text{cov,mf}}^{k,\text{l}}(K,\gamma)-P_{\text{cov,mf}}^{k+1,\text{l}}(K,\gamma) \nonumber\\
&=\mathbb{E}_{r_k,I_{r1}}\left[1-(1-e^{- \gamma r_k^{\alpha} I_{r1}})^L\right] \nonumber\\
&~~~-\mathbb{E}_{r_{k+1},I_{r3}}\left[1-(1-e^{- \gamma r_{k+1}^{\alpha} I_{r3}})^L\right] \nonumber\\
&=\mathbb{E}_{r_k,I_{r1},r_{k+1},I_{r3}}\left[(1-e^{- \gamma r_{k+1}^{\alpha} I_{r3}})^L-(1-e^{- \gamma r_k^{\alpha} I_{r1}})^L\right]  \nonumber\\
&\geq0,
\end{align}}
where the last step follows from that $r_{k+1}\geq r_k$ and $I_{r3}=\sum_{j\in\Phi_b\backslash\{\textbf{d}_k\}}g_{j,\text{mf}}\cdot r_j^{-\alpha}\geq I_{r1}=\sum_{j\in\Phi_b\backslash\{\textbf{d}_{k+1}\}}g_{j,\text{mf}}\cdot r_j^{-\alpha}$ since all $g_{j,\text{mf}}$ are i.i.d. random variables. Therefore, we conclude that the approximate probability in MF scheme is a non-increasing function of $k$. For the ZF scheme, the proof is similar to the MF scheme.

Since the coverage probabilities in both MF and ZF schemes are non-increasing functions of $k$, the second order derivative of the objective function (\ref{eqn:T3}) respect to $a_n$ can be expressed as:
{
\allowdisplaybreaks[4]
\begin{align}
&\frac{\partial^2 \sum_{n=1}^Np_nL(a_n)}{\partial a_n^2}  \nonumber\\
&=\sum_{n=1}^N p_n{\sum_{k=1}^K(k-1)(1-a_n)^{k-3}(ka_n-2)P_{\text{cov}}^k(K,\gamma)} \nonumber\\
&=\sum_{n=1}^N p_n\Bigg[-2P_{\text{cov}}^2(K)+2(3a_n-2)P_{\text{cov}}^3(K,\gamma) \nonumber\\
&~~~+\sum_{k=4}^K(k-1)(1-a_n)^{k-3}(ka_n-2)P_{\text{cov}}^k(K,\gamma)\Bigg] \nonumber\\
&\leq\sum_{n=1}^N p_n\Bigg[-2P_{\text{cov}}^3(K)+2(3a_n-2)P_{\text{cov}}^3(K,\gamma)  \nonumber\\
&+\sum_{k=4}^K(k-1)(1-a_n)^{k-3}(ka_n-2)P_{\text{cov}}^k(K,\gamma)\Bigg] \nonumber\\
&=\sum_{n=1}^N p_n\Bigg[6(a_n-1)P_{\text{cov}}^3(K,\gamma)+3(1-a_n)(4a_n-2)P_{\text{cov}}^4(K,\gamma) \nonumber\\
&+\sum_{k=5}^K(k-1)(1-a_n)^{k-3}(ka_n-2)P_{\text{cov}}^k(K,\gamma)\Bigg] \nonumber\\
&{~~~~~\vdots}  \nonumber\\
&\leq\sum_{n=1}^N p_n\left[ K(K-1)(1-a_n)^{K-3}(a_n-1)P_{\text{cov}}^K(K,\gamma)\right] \nonumber\\
&\leq 0,
\end{align}}
{where the previous steps come from the property that $P_{\text{cov}}^k(K,\gamma)$ is non-increasing of $k$ and the last step follows from that $0\leq a_n\leq 1$.}
Since objective of $\textbf{P1}$ is to maximize a concave function and all constraints are linear, $\textbf{P1}$ is a convex problem in terms of AFOT maximization.

Since $R_k(K)$ is also a non-increasing function of $k$, the objective function AESE can be proved to be concave similar to AFOT above. Therefore, the proof is completed.
\begin{center}
Appendix E: Proof of Theorem 3
\end{center}

For AFOT maximization, the Lagrangian function of $\textbf{P1}$ can be written as:
\begin{align}
L(a_1,a_2,\cdots,a_N,\mu)&=\sum_{n=1}^N p_n{\sum_{k=1}^K a_n(1-a_n)^{k-1}P_{\text{cov}}^{k}(K,\gamma)} \nonumber\\
&+\mu\left(M-\sum_{n=1}^Na_n\right),
\end{align}
where $\mu$ is the Lagrangian multiplier associated with the constraint (\ref{eqn:constraint}).
By letting the partial derivative of the Lagrangian function to be $0$, we have
\begin{align}
p_n \sum_{k=1}^K(1-a_n)^{k-2}(1-ka_n)P_{\text{cov}}^k(K,\gamma)=\mu \label{eqn: Lagrangian}
\end{align}

It is easy to find that the left hand of (\ref{eqn: Lagrangian}) is a decreasing function of $a_n$ since the objective function is concave. Notice that we have the constraint $0\leq a_n \leq 1$. Thus, when $a_n=1$, $\mu$ has the minimum value: $p_n\left[P_{\text{cov}}^1(K,\gamma)-P_{\text{cov}}^2(K,\gamma)\right]$. While for $a_n=0$, it has the maximum value: $p_n\sum_{k=1}^KP_{\text{cov}}^k(K,\gamma)$.
Therefore, the cache solution $a_n(\mu)$ is given by:
\begin{align}
a_n(\mu)=
\begin{cases}
1,&\mu \leq p_n\left[P_{\text{cov}}^1(K,\gamma)-P_{\text{cov}}^2(K,\gamma)\right] \\
w_n(\mu),&\text{otherwise}\\
0,&\mu \geq p_n\sum_{k=1}^KP_{\text{cov}}^k(K,\gamma)
\end{cases}, \label{eqn:strategy}
\end{align}
which is equivalent to (\ref{eqn:11111}) by substituting $\mu^*$ for $\mu$ in (\ref{eqn:strategy}). Hence, the proof is completed.

For AESE maximization, the proof is similar and hence is omitted here.
\end{appendices}
\bibliographystyle{IEEEtran}
\bibliography{IEEEabrv,paper1}

\end{document}